\documentclass[a4paper]{amsart}
\usepackage[a4paper,margin=2.5cm,includefoot]{geometry}
\usepackage[utf8]{inputenc}
\usepackage[english]{babel}
\usepackage[T1]{fontenc}

\usepackage{amsmath,amssymb,mathtools,bm}
\usepackage{csquotes}
\usepackage{amstext}
\usepackage{amsthm}
\usepackage{bbm}
\usepackage{enumerate}
\usepackage{hyperref}
\usepackage{breakurl}
\usepackage[nameinlink,capitalise]{cleveref}
\usepackage{braket}
\usepackage{dsfont}
\usepackage{comment}
\usepackage{color}
\usepackage{tikz}
\usepackage{pgfplots}\pgfplotsset{compat=1.8}

\makeatletter

\numberwithin{equation}{section}
\newtheorem{thm}{Theorem}[section]
\newtheorem{lem}[thm]{Lemma}
\newtheorem{prop}[thm]{Proposition}
\newtheorem{cor}[thm]{Corollary}

\theoremstyle{definition}

\newtheorem{hyp}{Hypothesis}
\renewcommand*{\thehyp}{\Alph{hyp}}

\theoremstyle{remark}
\newtheorem{rem}[thm]{Remark}
\newtheorem{ex}[thm]{Example}

\crefname{hyp}{Hypothesis}{Hypotheses}
\Crefname{hyp}{Hypothesis}{Hypotheses}
\crefname{lem}{Lemma}{Lemmas}
\Crefname{lem}{Lemma}{Lemmas}
\crefname{thm}{Theorem}{Theorems}
\Crefname{thm}{Theorem}{Theorems}
\crefname{prop}{Proposition}{Propositions}
\Crefname{prop}{Proposition}{Propositions}
\crefname{cor}{Corollary}{Corollaries}
\Crefname{cor}{Corollary}{Corollaries}
\crefname{enumi}{}{}
\Crefname{enumi}{}{}
\creflabelformat{enumi}{#2(#1)#3}
\crefname{equation}{}{}
\Crefname{equation}{}{}
\crefname{rem}{Remark}{Remarks}
\Crefname{rem}{Remark}{Remarks}

\makeatletter
\renewcommand{\@upn}{} 
\makeatother

\usepackage[inline]{enumitem}
\newlist{enumthm}{enumerate}{1} 
\setlist[enumthm]{label=\upshape(\roman*),ref=\thethm~(\roman*)}  
\crefalias{enumthmi}{thm} 
\newlist{enumcor}{enumerate}{1}
\setlist[enumcor]{label=\upshape(\roman*),ref=\thecor~(\roman*)}
\crefalias{enumcori}{cor}
\newlist{enumlem}{enumerate}{1}
\setlist[enumlem]{label=\upshape(\roman*),ref=\thelem~(\roman*)}
\crefalias{enumlemi}{lem}
\newlist{enumprop}{enumerate}{1}
\setlist[enumprop]{label=\upshape(\roman*),ref=\theprop~(\roman*)}
\crefalias{enumpropi}{prop}
\newlist{enumhyp}{enumerate}{1}
\setlist[enumhyp]{label=\upshape(\roman*),ref=\thehyp~(\roman*)}
\crefalias{enumhypi}{hyp}
\newlist{enumproof}{enumerate*}{1}
\setlist[enumproof]{label=\upshape(\roman*)}
\newlist{enumdef}{enumerate}{1}
\setlist[enumdef]{label=\upshape(\roman*),ref=\thedefn~(\roman*)}
\crefalias{enumdefi}{defn}
\newlist{enumrem}{enumerate}{1}
\setlist[enumrem]{label=\itshape(\alph*),ref=\therem~(\alph*)}
\crefalias{enumrem}{rem}

\makeatletter
\newcounter{subcreftmpcnt} %
\newcommand\romansubformat[1]{(\roman{#1})} 
\def\subcref{\@ifstar\@@subcref\@subcref}
\newcommand\@subcref[2][\romansubformat]{%
	\ifcsname r@#2@cref\endcsname
	\cref@getcounter {#2}{\mylabel}%
	\setcounter{subcreftmpcnt}{\mylabel}%
	\hyperref[#2]{\romansubformat{subcreftmpcnt}}%
	\else ?? \fi}   
\newcommand\@@subcref[2][\romansubformat]{%
	\ifcsname r@#2@cref\endcsname
	\cref@getcounter {#2}{\mylabel}%
	\setcounter{subcreftmpcnt}{\mylabel}%
	\romansubformat{subcreftmpcnt}%
	\else ?? \fi}   
\makeatother

\makeatletter
\DeclareRobustCommand{\crefnosort}[1]{%
	\begingroup\@cref@sortfalse\cref{#1}\endgroup
}
\makeatother

\def\endstepsymbol{$\lozenge$}
\def\endclaimsymbol{$\lozenge$}
\newcounter{proofstep}
\AtBeginEnvironment{proof}{\setcounter{proofstep}{0}}

\crefname{proofstep}{Step}{Steps}
\Crefname{proofstep}{Step}{Steps}
\newcounter{proofclaim}
\AtBeginEnvironment{proof}{\setcounter{proofclaim}{0}}

\crefname{proofclaim}{Claim}{Claims}
\Crefname{proofclaim}{Claim}{Claims}

\newcommand{\cC}{{\mathcal C}}
\newcommand{\cD}{{\mathcal D}}\newcommand{\cF}{{\mathcal F}}

\newcommand{\cS}{{\mathcal S}}
\newcommand{\cX}{{\mathcal X}}

\newcommand{\BC}{{\mathbb C}}

\newcommand{\BN}{{\mathbb N}}
\newcommand{\BR}{{\mathbb R}}

\usepackage{dsfont} 

\newcommand{\dsone}{{\mathds 1}}

\usepackage{mathrsfs}

\newcommand{\sfD}{{\mathsf D}}

\newcommand{\sfa}{{\mathsf a}}\def\sfb{{\mathsf b}}\newcommand{\sfc}{{\mathsf c}}
\newcommand{\sfd}{{\mathsf d}}
\newcommand{\sfi}{{\mathsf i}}
\newcommand{\sfl}{{\mathsf l}}
\newcommand{\sfm}{{\mathsf m}}\newcommand{\sfn}{{\mathsf n}}\newcommand{\sfo}{{\mathsf o}}

\newcommand{\sfx}{{\mathsf x}}

\newcommand{\rme}{{\mathrm e}}

\newcommand{\IN}{\BN}\newcommand{\IR}{\BR}\newcommand{\IC}{\BC}

\newcommand{\RR}{\BR}

\newcommand{\IE}{\mathbb{E}}

\newcommand{\eps}{\varepsilon}\newcommand{\ph}{\varphi}

\newcommand{\e}{\rme}\renewcommand{\i}{\sfi}\newcommand{\Id}{\dsone} \renewcommand{\d}{\sfd}

\renewcommand{\Re}{\operatorname{Re}}\renewcommand{\Im}{\operatorname{Im}}

\newcommand{\supp}{\operatorname{supp}}

\newcommand{\wh}[1]{\widehat{#1}}\newcommand{\wt}[1]{\widetilde{#1}}

\DeclareFontFamily{U}{mathx}{\hyphenchar\font45}
\DeclareFontShape{U}{mathx}{m}{n}{
	<5> <6> <7> <8> <9> <10>
	<10.95> <12> <14.4> <17.28> <20.74> <24.88>
	mathx10
}{}
\DeclareSymbolFont{mathx}{U}{mathx}{m}{n}
\DeclareFontSubstitution{U}{mathx}{m}{n}
\DeclareMathAccent{\widecheck}{0}{mathx}{"71}
\DeclareMathAccent{\wideparen}{0}{mathx}{"75}

\DeclareFontFamily{OMX}{MnSymbolE}{}
\DeclareFontShape{OMX}{MnSymbolE}{m}{n}{
	<-6>  MnSymbolE5
	<6-7>  MnSymbolE6
	<7-8>  MnSymbolE7
	<8-9>  MnSymbolE8
	<9-10> MnSymbolE9
	<10-12> MnSymbolE10
	<12->   MnSymbolE12}{}
\DeclareSymbolFont{mnlargesymbols}{OMX}{MnSymbolE}{m}{n}
\SetSymbolFont{mnlargesymbols}{bold}{OMX}{MnSymbolE}{b}{n}
\DeclareMathDelimiter{\llangle}{\mathopen}{mnlargesymbols}{'164}{mnlargesymbols}{'164}
\DeclareMathDelimiter{\rrangle}{\mathclose}{mnlargesymbols}{'171}{mnlargesymbols}{'171}
\DeclareMathDelimiter{\lsem}{\mathopen}{mnlargesymbols}{'102}{mnlargesymbols}{'102}
\DeclareMathDelimiter{\rsem}{\mathclose}{mnlargesymbols}{'107}{mnlargesymbols}{'107}
\DeclareMathDelimiter{\langlebar}{\mathopen}{mnlargesymbols}{'152}{mnlargesymbols}{'152}
\DeclareMathDelimiter{\ranglebar}{\mathclose}{mnlargesymbols}{'157}{mnlargesymbols}{'157}
\DeclareMathDelimiter{\lWavy}{\mathopen}{mnlargesymbols}{'137}{mnlargesymbols}{'137}
\DeclareMathDelimiter{\rWavy}{\mathopen}{mnlargesymbols}{'137}{mnlargesymbols}{'137}

\newcommand{\chr}{\mathbf 1}
\newcommand{\abs}[1]{\lvert#1\lvert}\newcommand{\Abs}[1]{\left\lvert#1\right\lvert}
\newcommand{\norm}[1]{\lVert#1\lVert}\newcommand{\Norm}[1]{\left\lVert#1\right\lVert}

\newcommand{\FGamma}{\Gamma}
\newcommand{\FS}{\cF}\newcommand{\dG}{\sfd\FGamma}\newcommand{\ad}{a^\dagger}

\renewcommand{\sc}[1]{\braket{#1}}
\newcommand{\nn}[1]{\Norm{#1}}

\renewcommand{\:}{\colon}

\newcommand{\dist}{\operatorname{dist}}

\title{On Lieb--Robinson bounds for a class of continuum fermions}
\date{November 10, 2023}
\author{Benjamin Hinrichs}
\address{Benjamin Hinrichs, Universit\"at Paderborn, Institut f\"ur Mathematik, Institut f\"ur Photonische Quantensysteme, Warburger Str. 100, 33098 Paderborn, Germany}
\email{benjamin.hinrichs@math.upb.de}

\author{Marius Lemm}
\author{Oliver Siebert}
\address{Marius Lemm and Oliver Siebert, Universit\"at T\"ubingen, Fachbereich Mathematik, Auf der Morgenstelle 10, 72076 T\"ubingen, Germany}
\email{marius.lemm@uni-tuebingen.de, oliver.siebert@uni-tuebingen.de}

\usepackage{xcolor}\usepackage{cancel}

\definecolor{green}{rgb}{0.0, 0.5, 0.5}
\definecolor{yellow}{rgb}{0.5, 0.5, 0}
\definecolor{lgray}{gray}{0.9}
\definecolor{llgray}{gray}{0.95}
\definecolor{lllgray}{gray}{0.975}

\definecolor{darkcerulean}{rgb}{0.03, 0.27, 0.49} 	\definecolor{darkcoral}{rgb}{0.8, 0.36, 0.27}
\usepgfplotslibrary{colormaps}
\pgfplotsset{/pgfplots/colormap/bluewhitered/.style={/pgfplots/colormap={bluewhitered}{color(0cm)=(darkcerulean); color(0.5cm)=(white); color(1cm)=(darkcoral)}}}
\pgfplotsset{/pgfplots/colormap/bluered/.style={/pgfplots/colormap={bluered}{color(0cm)=(darkcerulean); color(1cm)=(darkcoral)}}}
\definecolor{darkgreen}{rgb}{0,.39,0}
\usepackage[outline]{contour}

\usepgfplotslibrary{external}
\tikzset{external/system call={pdflatex \tikzexternalcheckshellescape -extra-mem-top=10000000 -extra-mem-bot=10000000 -halt-on-error -interaction=batchmode -jobname "\image" "\texsource"}}

\newcommand{\nmax}{n_{\sfm\sfa\sfx}}
\newcommand{\chir}[1]{\chr_{\{\abs{\cdot}<#1\}}}
\newcommand{\chiR}[1]{\chr_{\{\abs{\cdot}\ge#1\}}}
\newcommand{\chirx}[2]{\chr_{\{\abs{\cdot - #2}<#1\}}}
\newcommand{\chiRx}[2]{\chr_{\{\abs{\cdot - #2}\ge#1\}}}
\newcommand{\bx}[2]{x_{#1,#2}}
\newcommand{\cw}{c_W}
\newcommand{\Af}{\mathcal{A}}
\newcommand{\tr}{{\operatorname{tr}}}
\newcommand{\Ad}{\operatorname{Ad}}

\newcommand{\ttls}{\tau_t^{\Lambda}}
\newcommand{\ttzs}{\tau_t^{0}}

\newcommand{\gs}{\psi_\Lambda}
\newcommand{\gse}{E_\Lambda}
\newcommand{\gsP}{P_\Lambda}
\newcommand{\dmin}{d_{\sfm\sfi\sfn}}

\newcommand{\Cob}[1]{C^{#1}_{\mathrm{ob}}}
\newcommand{\Cmb}[1]{C_{#1,\mathrm{mb}}}
\newcommand{\Cobt}[1]{\tilde C^{#1}_{\mathrm{ob}}}
\newcommand{\Ctailx}[1]{C^{#1}_{\mathrm{pos}}}
\newcommand{\CtailE}[1]{C^{#1}_{\mathrm{En}}}
\newcommand{\fXc}[1]{f_{#1}}
\newcommand{\hk}{\mathfrak{K}}
\newcommand{\Ximb}{\Xi_{\mathrm{mb}}}
\newcommand{\Nf}{\mathcal{N}}

\usepackage[style=alphabetic,giveninits=true,maxnames=5,maxalphanames=4, backend=biber,url=true,isbn=false,eprint=false]{biblatex}
\renewbibmacro{in:}{%
	\ifentrytype{article}{}{\printtext{\bibstring{in}\intitlepunct}}}
\newbibmacro{string+doi}[1]{%
	\iffieldundef{doi}{#1}{\href{https://doi.org/\thefield{doi}}{#1}}}
\DeclareFieldFormat{title}{\usebibmacro{string+doi}{\mkbibemph{#1}}}
\DeclareFieldFormat[article]{title}{\usebibmacro{string+doi}{\mkbibquote{#1}}}
\DeclareFieldFormat[inproceedings]{title}{\usebibmacro{string+doi}{\mkbibquote{#1}}}
\DeclareFieldFormat[incollection]{title}{\usebibmacro{string+doi}{\mkbibquote{#1}}}
\DeclareFieldFormat[unpublished]{title}{\usebibmacro{string+doi}{\mkbibquote{#1}}}

\begin{document}

\begin{abstract} 
	\noindent
	We consider the quantum dynamics of a many-fermion system in $\IR^d$ with an ultraviolet regularized pair interaction as previously studied in [M. Gebert, B. Nachtergaele, J. Reschke, and R. Sims, Ann. Henri Poincar\'e 21.11 (2020)]. We provide a Lieb--Robinson bound under substantially relaxed assumptions on the potentials. We also improve the associated one-body Lieb--Robinson bound on $L^2$-overlaps to an almost ballistic one (i.e., an almost linear light cone) under the same relaxed assumptions.
	Applications include the existence of the infinite-volume dynamics and clustering of ground states in the presence of a spectral gap. We also develop a fermionic continuum notion of conditional expectation and use it to approximate time-evolved fermionic observables by local ones, which opens the door to other applications of the Lieb--Robinson bounds.
\end{abstract}

\maketitle

\section{Introduction}
\noindent
A Lieb--Robinson bound (LRB) \cite{LiebRobinson.1972} establishes an upper bound on the propagation speed of quantum information in a quantum many-body system. These bounds are typically ``ballistic'' in the sense that they bound the propagation speed uniformly in time. As first discovered by Hastings in the early 2000s LRBs are decisive analytical tools for resolving fundamental problems in quantum many-body physics and quantum information theory; see \cite{hastings2004lieb,hastings2005quasiadiabatic,BravyiHastingsVerstraete.2006,HastingsKoma.2006,NachtergaeleSims.2006,NachtergaeleOgataSims.2006,hastings2007area} and the later reviews \cite{nachtergaele2010lieb,gogolin2016equilibration,chen2023speed}. LRBs and their applications are mainly confined to the setting of lattice Hamiltonians with local and bounded interactions (typically quantum spin systems) with recent extensions to long-range spin interactions  \cite{chen2019finite,kuwahara2020strictly,tran2020hierarchy,tran2021lieb} and lattice fermions \cite{gluza2016equilibration,NachtergaeleSimsYoung.2018}. However, the standard techniques for deriving LRBs break down for \textit{unbounded} interactions. A natural class of lattice systems with unbounded interactions are lattice bosons and extending LRBs to these has been an active topic recently \cite{faupin2022maximal,faupin2022lieb,yin2022finite,kuwahara2022optimal,van2023optimal,lemm2023information,sigal2023propagation,lemm2023microscopic}.
%

For quantum many-body systems in \textit{continuous space}, another challenge for deriving LRBs comes from the ultraviolet divergences at short distances which are associated with unbounded energy and hence unbounded propagation speed. For continuum few-body systems (e.g., defined by linear Schr\"odinger operators $-\Delta+V$ on $L^2(\mathbb R^d)$ or variants thereof) one can derive propagation bounds through energy cutoffs \cite{SigalSoffer.1988,Skibsted.1991,HunzikerSigalSoffer.1999,ArbunichPusateriSigalSoffer.2021,BreteauxFaupinLemmSigal.2022,breteaux2023light} and this also applies to the nonlinear Hartree equation \cite{arbunich2023maximal} that emerges from quantum many-body dynamics in suitable scaling limits (e.g., mean-field) \cite{elgart2004nonlinear,elgart2006gross}.

So far, only a small number of works in mathematical physics have considered propagation bounds for \textit{continuum quantum many-body Hamiltonians} that are uniform in the total particle number and hold robustly outside of special parameter regimes (e.g., scaling limits). Bony, Faupin, and Sigal \cite{bony2012maximal} controlled the propagation speed of photons in a non-relativistic QED model. 

The present paper is motivated by a work of Gebert, Nachtergaele, Reschke, and Sims \cite{GebertNachtergaeleReschkeSims.2020} who studied a  class of continuum fermions whose pair interaction is smeared out (``UV-regularized'') with a fixed Gaussian $\ph$, i.e., they considered the many-body Hamiltonian
\begin{equation}\label{eq:H}
H_\Lambda=  \int \left(\nabla a_x^\dagger \nabla a_x+V(x) a_x^\dagger a_x\right)\d x+ \int_{\Lambda}\int_{\Lambda}  W(x-y)\ad(\ph_x)\ad(\ph_y)a(\ph_y)a(\ph_x) \d x \d y
\end{equation}
on the fermionic Fock space $\mathcal F_{\mathrm{antisymm}}(L^2(\mathbb R^d))$. For lattice fermions, a typical Lieb--Robinson bound would control the anticommutator $\|\{\tau_t(a(f)),\ad(g)\}\|$ where $\tau_t$ is the Heisenberg time evolution (conjugation with $e^{-\mathrm{i}tH_\Lambda}$) and $f,g\in L^2(\mathbb R^d)$.
The authors of \cite{GebertNachtergaeleReschkeSims.2020} instead focused on controlling the difference of anticommutators where one is evolved with the full dynamics $\tau_t$ and the other one is evolved with the non-interacting dynamics $\tau_t^0$ (where the Hamiltonian is taken to be \eqref{eq:H} with $W=0$). That is, they consider the quantity
\begin{equation}\label{eq:quantity}
\|\{\tau_t(a(f)),\ad(g)\} -\{\tau_t^0(a(f)),\ad(g)\}\|
=\|\{\tau_t(a(f)),\ad(g)\} -\langle e^{-\mathrm{i}t(-\Delta+V)},g\rangle\}\|
\end{equation}
where the equality holds by the CAR. The quantity \eqref{eq:quantity} \textit{isolates the effect of the interactions on quantum propagation}. This quantity can in principle be bounded without using any UV regularization for the one-body part in the Hamiltonian $H_\Lambda$.
Indeed, the main result of \cite{GebertNachtergaeleReschkeSims.2020} is a Lieb--Robinson bound (LRB) on this quantity
\begin{equation}\label{eq:GNRS}
\|\{\tau_t(a(f)),\ad(g)\} -\{\tau_t^0(a(f)),\ad(g)\}\|
\leq \|f\|_1 \|g\|_1 e^{C(t)-a(t)  \dist(\supp f,\supp g)}
\end{equation}
where $f,g\in L^1(\mathbb R^d)\cap L^2(\mathbb R^d)$ and $C(t)$ and $1/a(t)$ are both explicit and polynomially growing in $t$. A key application of the bound \eqref{eq:GNRS} is to derive existence of the thermodynamic limit of the dynamics of $\lim_{\Lambda\to\infty}\tau_t(a(f))$ \cite{GebertNachtergaeleReschkeSims.2020}. Notice that for this application the growth of $C(t)$ and $a(t)$ is unproblematic since $t$ is held fixed as $\Lambda \to\infty$. The proof of \eqref{eq:GNRS} in \cite{GebertNachtergaeleReschkeSims.2020} required the strong assumptions on $V$ and $W$ appearing in \eqref{eq:H}: (i) $V$ is assumed to be the Fourier transform of a compactly supported even signed measure (which implies that it extends to an entire function by the Paley-Wiener theorem) and (ii) $W$ is exponentially decaying. This raised the open problem to derive a many-body LRB under substantially weaker assumption on the potentials $V$ and $W$.
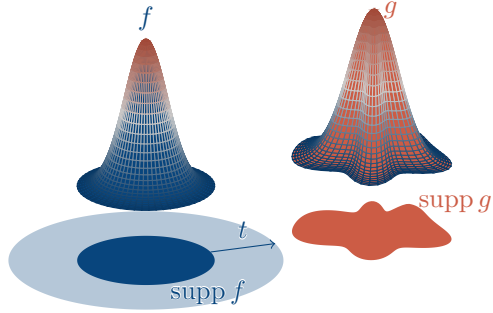
\begin{figure}
	\begin{center}
		\pgfmathdeclarefunction{myfunctx}{1}{%
			\pgfmathparse{(1+cos(#1*(#1-360)*(#1-90)/6000)/5)*sin(#1)}%
		}
		\pgfmathdeclarefunction{myfuncty}{1}{%
			\pgfmathparse{(1+cos(#1*(#1-360)*(#1-90)/6000)/5)*cos(#1)}%
		}
		\begin{tikzpicture}
			\pgfplotsset{colormap/bluewhitered}
			\begin{axis}[
				zmin=0,
				ticks=none,
				clip=false,
				axis line style={draw=none}
				]
				\addplot[fill=darkcerulean,opacity=.3,domain=0:360,samples=200,draw=none] ({2*cos(x)},{2*sin(x)});
				\addplot[darkcerulean,->,domain={1/sqrt(2)}:{sqrt(2)}] {x} node[midway,above] {\contour{white}{$t$}};
				\addplot[fill=darkcerulean,domain=0:360,samples=200,draw=none] ({cos(x)},{sin(x)}) node[below] {\contour{white}{$\textcolor{darkcerulean}{\supp f}$}};
				\addplot[fill=darkcoral,domain=0:360,samples=200,draw=none] ({2.5+myfunctx(x)},{2.5+myfuncty(x)}) node[right] {\contour{white}{$\textcolor{darkcoral}{\supp g}$}};
				\addplot3[surf,fill=darkcerulean,domain=0:360,y domain=0:1,variable = \t, variable y = \s, samples = 60, samples y = 30] ({s*sin(t)},{s*cos(t)},{.5+exp(-5*s^2)}) node[above]{\contour{white}{$\textcolor{darkcerulean}{f}$}};
				\addplot3[surf,fill=darkcoral,domain=0:360,y domain=0:1,variable = \t, variable y = \s,samples=60, samples y = 30 ] ({2.5+s*myfunctx(t)},{2.5+s*myfuncty(t)},{.5+exp(-5*s^2)}) node[right] {\contour{white}{$\textcolor{darkcoral}{g}$}};
			\end{axis}
		\end{tikzpicture} 
			\caption{Schematic for applying the Lieb--Robinson bounds \eqref{eq:GNRS}--\eqref{eq:obLR}: For sufficiently regular $f$ and $g$, the many-body anti-commutator $\|\{\tau_t(a(f)),a(g)\}\|$ and the one-body overlap $\langle e^{-\mathrm{i}t (-\Delta+V)}f,g\rangle$ are both small for $\abs t\ll \dist(\supp f,\supp g)$.}\label{fig:LRbound}
	\end{center}
\end{figure}

In this paper, we prove a slightly different Lieb--Robinson bound on the quantity \eqref{eq:quantity} (\cref{thm:manybody})
\begin{equation}\label{eq:mbLR}
\|\{\tau_t(a(f)),\ad(g)\} -\{\tau_t^0(a(f)),\ad(g)\}\| \le \|f\|_2\|g\|_2  e^{C_n(t)} \dist(\supp f,\supp g)^{-n},
\end{equation}
which only requires  $f,g\in L^2(\mathbb R^d)$ and holds under substantially weaker regularity assumptions on $V$ and $W$. Essentially we need that $V\in C_\sfb^{2n}$ and $W(x)\lesssim |x|^{-n}$; see \cref{hyp:V} and  \cref{hyp:W} for the precise formulations.  This result addresses the above-mentioned  problem to obtain an LRB also under substantially relaxed assumptions on $V$ and $W$.  As a corollary, we obtain existence of the thermodynamic limit of the dynamics for our much broader class of potentials.

\medskip
A new dynamical concept introduced in  \cite{GebertNachtergaeleReschkeSims.2020} is that of a \textit{one-body Lieb--Robinson bound} of the form
\begin{align}\label{eq:obLR}
	\abs{\braket{f,\e^{-\i t (-\Delta + V)} \ph}}^p \le \int F(t,\abs{y})\abs{f(y)}^p\d y \quad\mbox{for all}\ f\in L^2(\IR^d)\ \mbox{and some}\ p\ge 1, 
\end{align}
where $F(t,R)$ decays rapidly in $R$ outside of a spacetime region. This is leveraged into the many-body Lieb--Robinson bound \eqref{eq:GNRS} by a Duhamel-type iteration procedure. It is worth pointing out that this connection between one-body LRB and many-body LRB is fundamentally different from the case of integrable models where one-body transport fully characterizes the one-body LRB \cite{hamza2012dynamical,damanik2014new,damanik2015quantum,damanik2016anomalous,kachkovskiy2016transport,gebert2016polynomial}: here the model is truly interacting and the one-body LRB is just an estimate used in deriving the many-body LRB. 

In \cite{GebertNachtergaeleReschkeSims.2020}, \eqref{eq:obLR} is proved with $F(t,R)$ decaying exponentially in the spacetime region $R>t^3$ again assuming that $V$ is entire. This condition would allow for rapid acceleration. Another open problem has been to improve this to a linear light cone condition $R>vt$ as in the usual LRB. 
 
 In our other main result (\cref{prop:LRint}), we derive a bound of the form \cref{eq:obLR} with $F(t,R)$ decaying polynomially in the spacetime region $R>t^{1+\tfrac{1+2\delta}{n}}$ for any $\delta>0$  and $V\in C_\sfb^{2n}$. That is, we establish the existence of an almost linear cone for one-body Lieb--Robinson bounds for sufficiently large $n$, which addresses the second open problem.

\medskip
Regarding \textit{methods}, we prove the one-body Lieb--Robinson bound by refining techniques for deriving propagation speed bounds for Schr\"odinger operators \cite{SigalSoffer.1988,Skibsted.1991,HunzikerSigalSoffer.1999,ArbunichPusateriSigalSoffer.2021}. These describe upper bounds for the propagation of a Schr\"odinger particle, when an explicit ultraviolet energy cutoff is in place.
They are related to the many-body ASTLO (adiabatic spacetime localization observables) method used for lattice boson Hamiltonians in \cite{faupin2022maximal,faupin2022lieb,sigal2023propagation,lemm2023information,lemm2023microscopic}.
Since we do not use a hard ultraviolet cutoff in energy or momentum for the many-body fermions, such bounds are not directly available.
However, by using the rapid decay of the test function $\ph$ in momentum space, we can employ a momentum cutoff together with Markov's inequality and then use bounded-energy propagation speed bounds to establish an almost linear light cone for the one-body LRB \eqref{eq:obLR}.

Our derivation of the many-body Lieb--Robinson bound is then an appropriate adaption of the iteration scheme established in \cite{GebertNachtergaeleReschkeSims.2020}. We have to take some care here due to the fact that we consider the case $p=2$ and not $p=1$ in \cref{eq:obLR}.
This allows us to consider all generators of the CAR algebra $\{a(f)\}$ with $f$ ranging over the full $L^2(\IR^d)$.
We further show that our Lieb--Robinson bound extends to commutators and anti-commutators of a dense set in the CAR algebra (\cref{cor:commutators} and \cref{cor:anti commutators}).

\medskip
We complement our Lieb--Robinson bound with a variety of \textit{applications}. For this, we adapt  the by now standard techniques for lattice models to the continuum case, which in some cases leads to new technical challenges as we describe below.

As mentioned before, a standard application of LRBs for quantum spin systems is to construct a strongly continuous group of automorphisms of the CAR algebra in the thermodynamic limit (\cref{th:infinite volume dynamics}). In the discrete case, such results were proved in \cite{Robinson.1976,BratteliRobinson.1996,NachtergaeleOgataSims.2006,NachtergaeleSchleinSimsStarrZagrebnov.2010,NachtergaeleVershyninaZagrebnov.2011}.
The argument we use is the same as in  \cite{GebertNachtergaeleReschkeSims.2020} and the only point is that using our LRB one obtains the result for a much broader class of potentials $V$ and $W$ in \eqref{eq:H}. 

A second standard application of LRBs for quantum spin systems is the exponential decay of correlations of ground states with a spectral gap (also known as clustering of correlations) \cite{HastingsKoma.2006,NachtergaeleSims.2006}. Building on  techniques from these works, we derive an analogous result in the continuous setting (\cref{th:exponential clustering}). In doing this, we have to bound the norm of the commutator $[\tau_t(a(f)),a(g)]$ evolved by the full many-body dynamics. In our setting, this requires two separate estimates: first we use the many-body LRB \cref{eq:mbLR} to reduce this to an estimate on the non-interacting dynamics $\tau_t^0$. Second, we bound the non-interacting dynamics recalling that $\{\tau_t^0(a(f)),a(g)\}=\langle e^{-\mathrm{i}t (-\Delta+V)}f,g\rangle$ by the CAR and then we use the one-body LRB to bound the overlap $\langle e^{-\mathrm{i}t (-\Delta+V)}f,g\rangle$.

In the study of discrete quantum systems, LRBs provide the possibility to approximate time evolved local observables by strictly local ones \cite{BravyiHastingsVerstraete.2006,NachtergaeleOgataSims.2006}. This is a crucial ingredient in applications of LRBs, e.g., proving stability of the spectral gap \cite{bravyi2010topological,michalakis2013stability} and performance control on quantum simulation \cite{kliesch2014lieb} and various information-theoretic protocols \cite{epstein2017quantum}. For discrete quantum spin systems, the approximation result that effectively replaces time evolved local observables by strictly local ones follows directly by representing the partial trace as a Haar-average over local unitaries  \cite{BravyiHastingsVerstraete.2006,NachtergaeleOgataSims.2006}. For lattice fermions the implementation of this idea is similar but more subtle \cite{NachtergaeleSimsYoung.2018}.
In our setting of continuum fermions, we derive an approximation result that effectively replaces time evolved local observables by strictly local ones. For this, we build on the construction for lattice fermions in \cite{NachtergaeleSimsYoung.2018} but in the continuum new difficulties arise because $L^2(\Lambda)$ is infinite-dimensional even for bounded $\Lambda$. This has the consequence that we need to control propagation of many observables simultaneously in a summable way. We address this by introducing a notion of ``partial partial trace (PPT)'': For the PPT we only consider the partial trace with respect to a finite subset of the orthonormal basis of each dyadic annulus and we allow the cardinality of the subset to grow with the dyadic scale. This construction and our bound on it (\cref{th:conditional expectation LR}) open the door to many other applications of LRBs now in the context of continuum fermions. For example, one may consider the stability of the spectral gap for a class of frustration-free continuum fermions similar to \cite{bravyi2010topological,michalakis2013stability}. The open gap stability problem for fermions was emphasized in \cite{GebertNachtergaeleReschkeSims.2020} and there it was also pointed out that a natural setting to consider would be to perturb a gapped non-interacting Hamiltonian by turning on an interaction $W$. For lattice fermions, this was done by renormalization group methods in \cite{de2019persistence}.

\subsection*{Organization of the paper}
In the next \cref{sec:results}, we introduce the model precisely and present our main results. In \cref{sec:onebody,sec:manybody}, we prove the one-body and many-body Lieb--Robinson bounds, respectively. Afterwards, in \cref{sec:applications}, we derive the applications discussed above.

\subsection*{Notation}
Throughout this article, we fix the following conventions:
\begin{itemize}
	\item $B_R(x) \coloneqq \{ y \in \IR^d : \abs{x-y} \leq R\}$,
	\item $a\vee b \coloneqq \max\{a,b\}$, $a\wedge b \coloneqq \min\{a,b\}$,
	\item $\braket t \coloneqq (1 + t^2)^{1/2}$ for $t \in \IR$,
	\item $ a \lesssim  f(\sigma)$ means there exists a constant $C>0$ independent of $\sigma$ such that $a\le C f(\sigma)$,
	\item  $\wh{p}^j = - \i \partial_j$ denotes the momentum operator,
	\item $\cC_\sfb^{n}(\IR)$ and $\cC_0^{n}(\IR)$, $n \in \IN \cup \{\infty\}$, denote the $n$ times continuously differentiable bounded and compactly supported functions on $\IR$,
	\item $\cS$ is the set of Schwartz functions on $\IR^d$,
	\item $[A,B]$ and $\{A,B\}$ denote the commutator and anticommutator between $A$ and $B$,
	\item	$C^*(A)$ for a set $A$ of bounded operators denotes the $C^*$-subalgebra generated by $A$.
\end{itemize}

\section{Model and main results}\label{sec:results}
\noindent
In this \lcnamecref{sec:results}, we present the model under consideration and describe our main results.

{We first describe our result on} one-body Lieb--Robinson bounds, i.e., overlap bounds of the form \cref{eq:obLR} for Schr\"odinger operators \begin{equation}
	\label{eq:onebody op}
	T \coloneqq -\frac{\Delta}{2}  + V
\end{equation}
We will work under the following \lcnamecref{hyp:Lbound} for the electrostatic potential $V:\IR^d\to\IC$.
\begin{hyp}\label{hyp:V}~
	\begin{enumhyp}
		\item\label{hyp:Lbound} Assume $V\in L^2_{\sfl\sfo\sfc}(\IR^d)$ is $\Delta$-bounded with bound $<1$, i.e., there exists $a\in[0,1)$ and $b\ge 0$ such that
		\[ \nn{V f} \le a\nn{\Delta f} + b\nn{f} \qquad \mbox{for all}\ f\in \cD(\Delta) = H^2(\IR^d). \]
		\item\label{hyp:Vdiff} $V$ is $n_V \geq 2$ times differentiable with essentially bounded derivatives. 
	\end{enumhyp}
\end{hyp}

If \cref{hyp:Lbound} is satisfied, then $T$
defines a selfadjoint lower-semibounded operator on $\cD(T)=\cD(\Delta)$, by the Kato--Rellich theorem.

In the following, we will usually assume $\ph$ to be a
Schwartz function in $\IR^d$ and for any $x \in \IR^d$ denote the shifted function by
\[
\ph_x := \ph(\cdot - x).
\]
A particular choice of $\ph$ will be the $L^1$-normalized Gaussian of variance $\sigma > 0$, i.e., 
\begin{equation}\label{def:gauss}
	\ph^\sigma (y) \coloneqq (\pi\sigma^2)^{-d/2}\exp\left(-y^2/2\sigma^2\right)\quad\mbox{for}\ y\in\IR^d.
\end{equation}	
The $L^1$-normalization is natural in the context of many-body theories, since it implies the convergence of the many-body dynamics to pointwise interactions, see \cref{rem:pointwise} for more details.
\begin{rem}
	Instead of Schwartz functions and Gaussians, it is possible to generalize our result to arbitrary functions which are sufficiently localized in space and energy space, cf. \cref{th:tail estimates}.  For example, one could consider  $\ph \in \chr_{(-\infty,E)}(T) \in  L^2(\IR^d)$, in which case the energy localization is obvious. Spatial localization can be proven using Agmon estimates, cf. \cite{Agmon.1982}.
\end{rem}
Our first main result is the following one-body Lieb--Robinson bound. 
Its proof  is given in \cref{sec:onebody}.
\begin{thm}[{One-body Lieb--Robinson bound}]\label{prop:LRint}
	Assume \cref{hyp:V} holds, 
	let $n \in\IN$, $\ph\in\cS$ 
	and let $\delta>0$.
	\begin{enumthm}
		\item Assume $2n \leq n_V$. Then there exists a constant $\Cob{0} > 0$  such that 
		\[ \Abs{\braket{f,\e^{-\i t T}\ph_x}}^2 \leq \Cob{0} 
	\braket{t}^{1+2\delta}\int \left(1 \wedge \frac{\braket t}{\abs{x-y}} \right)^{n} \abs{f(y)}^2 \d y 
	\]
		for all $x \in \IR^d$, $f\in L^2(\IR^d)$ and $t\in\IR$.
		\item 
		\label{it:onebody time zero}
		Assume $2n+2 \leq n_V$.  Then there exists a constant $\Cob{1} > 0$  such that
		\[ \Abs{\braket{f,(\e^{-\i t T} - \Id)\ph_x}}^2 \leq \Cob{1}  \abs t^2
		\braket{t}^{1+2\delta}\int \left(1 \wedge \frac{\braket t}{\abs{x-y}} \right)^{n} \abs{f(y)}^2 \d y 
		\]
		for all $x \in \IR^d$, $f\in L^2(\IR^d)$ and $t\in\IR$.
	\end{enumthm}
\end{thm}
\begin{rem}
	If $\ph = \ph^\sigma$ {is the Gaussian} \cref{def:gauss} with $\sigma \leq 1$, then we have {the explicit bound}
	\begin{align*}
		\Cob{j} \lesssim \sigma^{-d} \sigma^{-4(n(n+\delta)/2\delta+j)}
		\qquad\mbox{for}
		\ j\in\{0,1\}.
	\end{align*}
\end{rem}
\begin{rem}\label{rem:obLR}
 The light cone in above theorem can easily seen to be the region $\abs{x-y} \lesssim \braket t^{1+(1+2\delta)/n}$, which is almost linear, at least for large $n\in\IN$.
 In comparison, the light cone in \cite[Theorem 2.3]{GebertNachtergaeleReschkeSims.2020} is given by $\abs{x-y} \lesssim \braket t^3$.
 Furthermore, our result allows for a broad class of one-body potentials $V$ and arbitrary Schwartz functions $\ph$.
 This comes with the cost that the decay outside of the light cone is polynomial and not exponential, as in \cite{GebertNachtergaeleReschkeSims.2020}.
\end{rem}
\begin{rem}
	Instead of the non-relativistic setting \eqref{eq:onebody op}, one could also consider a massive relativistic Schrödinger operator, i.e.,
	\[
	T = \sqrt{-\Delta + 1} + V.  
	\]
	In this case a similar one-body LRB as \cref{prop:LRint} can be derived by combining our method with  \cite[Theorem 1.1]{BreteauxFaupinLemmSigal.2022} (setting the Kraus operators $W_j=0$ in the statement therein).
	In this situation an energy cutoff is not required because the free group velocity $\partial_k \sqrt{k^2+1}$ is bounded which makes the analysis shorter and the singular behavior for $\sigma \downarrow 0$ becomes better, since the term $\sigma^{-4(n(n+\delta)/2\delta)}$  originating from the energy cutoff, does not appear. However, {the estimates are not sufficient to show uniform bounds in $\sigma$, which is a considerably more difficult problem.}
\end{rem}
The second main result is an application of the one-body Lieb--Robinson bound \cref{prop:LRint} to a model of smeared-out interacting fermions, similarly to \cite[\textsection~4]{GebertNachtergaeleReschkeSims.2020}. 
We first give the complete definition of the many-body Hamiltonian \eqref{eq:H}.
The many-body interaction is given by a function $W$ satisfying the following constraints. 
\begin{hyp}\label{hyp:W}
	Let $W\in L^1(\IR^d)\cap L^\infty(\IR^d)$ satisfy $W(x)=W(-x)$ for a.e. $x\in\IR^d$ and assume there exist $n_W\in\IN$  and $\cw > 0$ such that
	\begin{align}
		\label{eq:W decay}
		\abs{W(x)} \leq \cw \left(1 \wedge \frac{1}{\abs{x}} \right)^{n_W}  
	\end{align}
	for all $x \in \IR^d$. 
\end{hyp}
To introduce the many-body model, let
\[ \FS \coloneqq \bigoplus_{k=0}^\infty L^2_\sfa(\IR^{d\cdot n}) \]
denote
the fermionic Fock Space, where the anti-symmetrization in each summand occurs over the $n$ $d$-dimensional variables.
Given a selfadjoint operator $A$ on $L^2(\IR^d)$, we define its second quantization $\dG(A)$ on $\FS$ as the usual selfadjoint closure of
\[
0 \oplus \bigoplus_{k=0}^\infty  \Id \otimes \cdots \otimes \underbrace{A}_{k-\text{th position}} \otimes \cdots \otimes \Id. 
\]
Further, let $a(f)$, $f\in L^2(\IR^d)$ denote the fermionic annihilation operator, i.e.,
\[ a(f)(\psi^{(n)})_{n\in\IN_0} \coloneqq {\sqrt{n+1}}\left(\int \overline {f(k)} \psi^{(n+1)}(k,\cdots)\d k\right)_{n\in\IN_0}.  \]
This defines a bounded operator on $\FS$ with $\nn{a(f)}=\nn f$ and, denoting its adjoint as $\ad(f)\coloneqq a(f)^*$, we have the usual canonical anticommmutation relations (CAR)
\[ \{a(f),a(g)\}=\{\ad(f),\ad(g)\}=0, \quad \{a(f),\ad(g)\}=\braket{f,g} \qquad\mbox{for all}\ f\in  L^2(\IR^d).\]
The many-body Hamiltonian for the smeared-out fermions on $\FS$ is now given by 
\begin{equation}
	\label{eq:manybody op}
	\boxed{H_\Lambda \coloneqq \dG(T) + \int_{\IR^d}\int_{\IR^d} \chr_{\Lambda\times \Lambda}(x,y) W(x-y)\ad(\ph_x)\ad(\ph_y)a(\ph_y)a(\ph_x) \d x \d y}
\end{equation}
For any bounded and measurable $\Lambda\subset \IR^d$, this defines a selfadjoint lower-semibounded operator on $\FS$.
\begin{rem}\label{rem:pointwise}
	At this point, it is worth commenting on our assumption of $L^1$-normalization for the Gaussians \cref{def:gauss}, since it might seem unnatural from a Hilbert space perspective.
	However, assuming that the limit $\sigma\downarrow 0$ should converge to point interactions between the fermions, the $L^1$-normalization is essential, cf. \cite[Thm.~8.14]{Folland.1999}.
	Other normalizations would imply convergence to different non-physical limits.
	We refer to \cite[Appendix A]{GebertNachtergaeleReschkeSims.2020} for a proof of the operator convergence (in strong resolvent sense) in the limit $\sigma\downarrow0$.
\end{rem}
We consider the Heisenberg time evolution with respect to the interacting Hamiltonian \eqref{eq:manybody op} and the free Hamiltonian $\dG(T)$. Explicitly, for a bounded operator $B$ on $\FS$, we write
\begin{equation}
	\ttls(B) \coloneqq \e^{\i t H_\Lambda} B\e^{-\i t H_\Lambda}, \quad \ttzs(B) \coloneqq \e^{\i t \dG(T)} B\e^{-\i t \dG(T)}  \qquad \mbox{for}\ t\in\IR.
\end{equation}
Following \cite{GebertNachtergaeleReschkeSims.2020}, let
\begin{equation}
	F_t^\Lambda (f,g) \coloneqq
	\norm{\{\ttls (a(f)) - \ttzs (a(f)) , \ad(g) \} }
	+ \norm{\{\ttls(\ad(f)),\ad(g)\}}
\end{equation}
denote the `difference' between the interacting and the free time evolution. Applying our one-body result \cref{prop:LRint}, we can now state our second main result, a Lieb--Robinson bound for $F_t^\Lambda (f,g)$. The general concept of our bound is illustrated in \cref{fig:LRbound}.
\begin{thm}[{Many-body Lieb--Robinson bound}]
	\label{thm:manybody}
	Assume \cref{hyp:V,hyp:W} hold.
	Let $\delta > 0$ and $n \in \IN$ such that $n \leq \frac{n_V}{2} \wedge n_W$. Then for all $f,g \in L^2(\IR^d)$, 
	there exist constants $\Cmb{1},\Cmb{2}$ depending on $\sigma,n,d,\delta$ and the latter one also on $W$, such that
	\begin{align*}
	F_t^\Lambda(f,g)^2 \leq \Ximb(t) \int_{\IR^d\times\IR^d} \left(1 \wedge \frac{\braket t}{\abs{x-y}} \right)^{n} \abs{f(x)}^2\abs{g(y)}^2 \d(x,y),
\end{align*}
where
\begin{align}
	\label{eq:ximb}
	\Ximb(t) :=  \Cmb{1} \abs t \braket{t}^{2(1+2\delta+d)} \exp(\Cmb{2} \braket{t}^{1+2\delta+3d} t^2 ).
\end{align}
	In the case of Gaussians, the $\sigma$-dependence for small $\sigma$ of the constants $\Cmb{1},\Cmb{2}$ is given as
	\begin{align*}
	\Cmb{1} &\l\sigma^{-8(n(n+\delta)/\delta)} \lesssim (\Cob{0})^2 \lesssim \sigma^{-2d} , \\
	\Cmb{2} &\lesssim \sigma^{-d} \Cob{0} \nn{\ph^\sigma}^4  \lesssim  \sigma^{-4d} \sigma^{-4(n(n+\delta)/\delta)} 
\end{align*}
	for all $\sigma \leq 1$.
\end{thm}
In the following, we will always assume \cref{hyp:V,hyp:W} to be satisfied and set $\nmax = \frac12 n_V\wedge n_W$.
Let us discuss two generalizatons of \cref{thm:manybody}  for observables with multiple creation and annihilation operators, in which case we can also estimate (higher-order) commutators.
\begin{cor}[Higher commutator estimate] 
	\label{cor:commutators}
	Let $N,M\in\IN$ such that $N$ or $M$ is even, and assume that the functions $f_1,\ldots,f_N,g_1,\ldots,g_M\in L^2(\IR^d)$ are normalized.
	Then for all $n \in \IN$ with $n\le\nmax$
	\begin{align*}
		&\nn{ \Big[ (\ttls - \ttzs)(a^\#(f_1) ) \cdots (\ttls - \ttzs)(a^\#(f_N) ), a^\#(g_1) \cdots a^\#(g_M) \Big]  } \\ &\leq  	 \sum_{i=1}^N \sum_{j=1}^M  \left(\Ximb(t) \int_{\IR^d\times\IR^d} \left(1 \wedge \frac{\braket t}{\abs{x-y}} \right)^{n} \abs{f_i(x)}^2\abs{g_j(y)}^2 \d(x,y)\right)^{1/2} ,
	\end{align*}
	where each $a^\#$ can be individually chosen as $a$ or $\ad$.
\end{cor}
\begin{proof}
	W.l.o.g. assume that $M$ is even.
	We prove that the commutators $[a_1\cdots a_N,b_1\cdots b_M]$ can be expanded into a sum of expressions of the form
	\begin{equation}\label{eq:expansionprototype}
		\pm a_1\cdots a_{i-1}b_1\cdots b_{j-1}\{a_i,b_j\}b_{j+1}\cdots b_Ma_{i+1}\cdots a_N .
	\end{equation}
	Estimating each summand with \cref{thm:manybody} then proves the statement.
	
	In the first step, we use the commutator product rule $[A_1 A_2,B] = A_1[A_2,B] + [A_1,B]A_2$ to expand into expressions of the form $a_1\cdots a_{i-1}[a_i,b_1\cdots b_M]a_{i+1}\cdots a_N$. Then, we iteratively apply the identities
	\begin{align}
	\label{eq:comm anticomm}
		&[A,B_1 B_2]  = \{A,B_1\}B_2 - B_1 \{A,B_2\}, \\
	\label{eq:anticomm comm}
		&\{A,B_1 B_2\} = [A,B_1 ] B_2 + B_1 \{A,B_2\},
	\end{align}
	where we use \cref{eq:comm anticomm} when $B_1$ and $B_2$ can be chosen with an odd number of factors and \cref{eq:anticomm comm} if $B_1$ can be chosen with an even and $B_2$ with an odd number of factors.
\end{proof}
\begin{cor}[Higher anticommutator estimate] 
	\label{cor:anti commutators}
	Let $M,N\in\IN$ such that $N$ and $M$ are odd, and assume that the functions $f_1,\ldots,f_N$, $g_1,\ldots,g_M\in L^2(\IR^d)$ are normalized.
	Then for all $n \in \IN$ with $n\le\nmax$
		\begin{align*}
		&\nn{ \Big\{  (\ttls - \ttzs)(a^\#(f_1) ) \cdots (\ttls - \ttzs)(a^\#(f_N) ), a^\#(g_1) \cdots a^\#(g_M) \Big\}  } \\ &\leq  	 \sum_{i=1}^N \sum_{j=1}^M \left( \Ximb(t) \int_{\IR^d\times\IR^d} \left(1 \wedge \frac{\braket t}{\abs{x-y}} \right)^{n} \abs{f_i(x)}^2\abs{g_j(y)}^2 \d(x,y) \right)^{1/2}. 
	\end{align*}
\end{cor}
\begin{proof}	
	The proof is similar to that of \cref{cor:commutators}.
	Explicitly, we first expand the left hand side of the anticommutator $\{a_1\cdots a_N,b_1\cdots b_N\}$ by repeatedly applying \cref{eq:anticomm comm} and \cref{eq:comm anticomm}, again taking care that one side of a commutator always has an even number of factors. Subsequently, we do the same with the right hand side to arrive at a sum of expressions of the form \cref{eq:expansionprototype}.
%
%
\end{proof}
\begin{rem}
	\label{rem:wrong time evolution}
	It would be more natural to consider $(\ttls - \ttzs)(a^\#(f_1) \cdots a^\#(f_N))$ instead of the expression used in \cref{cor:commutators,cor:anti commutators}, but a simple algebraic argument does not seem to be available here. However, such expressions can be estimated with a similar product rule expansion and the one-body bound \cref{prop:LRint}, if we assume that the $f_1, \ldots, f_N$ or $g_1, \ldots, g_M$ are Schwartz functions.
\end{rem}
Finally, we describe several applications of our main theorems. Let
\[
\Af := C^* ( a(f) : f \in L^2(\IR^d) )
\]
denote the CAR algebra of the full system. First, we show the existence of a dynamics on $\Af$ in the infinite-volume limit. This is a direct generalization of \cite[Theorem 2.7]{GebertNachtergaeleReschkeSims.2020}, where our methods allow for a larger class of admissible interactions $W$. 
\begin{cor}[Existence of infinite-volume dynamics]
\label{th:infinite volume dynamics}
 Assume that $\nmax>2d$. Then for all $t \in \IR$, $f \in L^2(\IR^d)$ and any increasing sequence $(\Lambda_k)$ of bounded subsets of $\IR^d$ such that $\bigcup_{k} \Lambda_k = \IR^d$, the limit
 \[
 \lim_{k\to\infty} \tau_t^{\Lambda_k}(a(f)) =: \tau_t(a(f)).
 \]
 with respect to the norm topology 
 exists with uniform convergence in $t$ on compact intervals of $\IR$. Furthermore, $(\tau_t)_{t \in \IR}$ forms a strongly continuous one-parameter group of automorphisms on $\Af$. 
\end{cor}
A common application of Lieb--Robinson-bounds is the proof of exponential clustering for gapped systems. In the following result, we show that we can use the proof strategy of \cite{NachtergaeleSims.2006} in combination with \cref{thm:manybody} to obtain a weak version of clustering.
\begin{thm}[Clustering]
	\label{th:exponential clustering}
	Let $\gse$ be the ground state energy of $H_\Lambda$ and assume the ground state is gapped,  i.e.,
	\begin{align*}
		\gamma := \sup \{ \delta > 0 : \operatorname{Spec}(H_\Lambda - \gse ) \cap (0,\delta) = \emptyset  \} > 0. 
	\end{align*}
	Suppose that $A = \sum_{i=1}^N A_i$ and $B = \sum_{j=1}^M B_j$, where each $A_i$ is a product of $n_i \in \IN$ creation/annihilation operators $a^\#(\ph_{x_k})$ and each $B_j$ is a product of $m_j  \in \{1,2\}$ creation/annihilation operators $a^\#(\ph_{f_l})$ with compactly supported $f_k$ such that
	\[
	\dmin := \inf_{k,l} \dist( x_k, \supp f_l ) > 0. 
	\]
	Then for each $n\le\nmax$ there exists a constant $C$ independent of $A$ and $B$ such that, for all $b > 0$,
	\begin{align*}
		\abs{\sc{\psi_0^\Lambda, A \tau_{\i b}(B) \psi_0^\Lambda } } &\leq  \frac{C N}{\sqrt \gamma \wedge 1} \exp  \left(  - \frac{\gamma \wedge 1 }{2}    (n \log (\dmin+1))^{1/4} - \frac{\gamma b^2  }{2 (n \log (\dmin+1))^{1/(3+3d)}}  \right) \\ \qquad &\times \sum_{i,j}  \nn{A_i} \nn{B_j}  \left( n_i m_j \right)  .
	\end{align*}
\end{thm}
\begin{rem}
	The subpolynomial decay in contrast to the exponential decay proven in \cite{NachtergaeleSims.2006} is due to the form of our Lieb--Robinson-bounds.
	This is due to the polynomial decay proven in \cref{prop:LRint,thm:manybody} in combination with the exponentially growing prefactors in $t$ in \eqref{eq:ximb}.
\end{rem}
\begin{rem}
	The restriction to at most two creation and annihilation operators for the $B_j$ is not strictly necessary but allows for a more concise proof. Taking into account more factors would require to consider further commutators and differences, since we can only estimate terms of the form
	\[
	 \nn{ [ (\ttls - \ttzs)(a^\#(A) ), B] }
	\]
	for $A$ being a polynomial in the creation and annihilation operators of at most one, cf. \cref{rem:wrong time evolution}.
\end{rem}
Finally, to study the localization of observables, we introduce a conditional expectation similar to \cite{NachtergaeleSimsYoung.2018}.
For any measurable set $X \subseteq \IR^d$, we define a linear projection 
\begin{align*}
	\IE_X \: \Af \rightarrow \Af_X,
\end{align*}
where $\Af_X = C^*(\{ a(f) : f \in L^2(X)  \}) \subseteq \Af$ denotes the $C^*$-subalgebra of the CAR algebra.
We can in fact allow more general subalgebras given by closed subspaces $\hk \subseteq L^2(X)$, cf. \cref{sec:conditional expectation}, but restrict to the simpler case here for simplicity. 
If we restrict to the $C^*$-subalgebra $\Af^+ \subseteq \Af$ of even parity, we prove that $\IE_X$ indeed acts a conditional expectation, i.e., 
\[
\IE_\hk (BAC) = B \IE_\hk(A) C.
\]
for all $A \in \Af^+$ and $B,C \in \Af_X^+$, see \cref{th:tomiyama coro}. 

Furthermore, if we restrict to selected orthonormal subbases of $L^2(X^\sfc)$, having a limited number of modes on annuli around the set $X$, we can apply our Lieb--Robinson bound \cref{thm:manybody} to obtain
	\begin{align*}
		\nn{\ttls(A) - \IE_X(\ttls(A))} \leq  C(t) \left( \frac{\braket t}{ C_X \dist(X^\sfc,Y)} \right)^n
\end{align*}
for any observable $A\in\Af_Y$, where $C(t)$ denotes a constant growing in $t$.
The full statement is given in \cref{th:conditional expectation LR}. 


\section{From Propagation Bounds to One-Particle Lieb--Robinson Bounds}
\label{sec:onebody}
\noindent
Throughout this section, we assume that $V$ satisfies \cref{hyp:Lbound} without further mentioning and write
$  T \coloneqq -\Delta + V$.
To obtain a Lieb--Robinson bound for the one-particle Schr\"odinger operator $T$, we apply the following propagation speed bound which is similar to the one proven in \cite{ArbunichPusateriSigalSoffer.2021} {but with an explicit dependence on the energy}.
\begin{prop}\label{prop:APSS}
	Assume that $\xi\in\cC^\infty(\IR;[0,1])$ satisfies $\xi(x)=1$ for $x\le0$ and $\xi(x)=0$ for $x\ge1$. Finally, fix $\alpha>1$
	and let
	\begin{equation}\label{def:gE} g_E(x) \coloneqq \xi((x-E)/((\alpha-1) E)) \quad \mbox{and} \quad \sfc_E \coloneqq \norm{\abs{\wh p} \chr_{(-\infty,\alpha^2 E] }(T)},\ \wh p\coloneqq -\i\nabla \qquad \mbox{for}\ E\in\IR. \end{equation}
	Then, for all $n\in\IN$ and $\eps>0$, there exists $C_{n,\eps}>0$, otherwise only depending on $V$, $\xi$ and $\alpha$,
	such that for all $E>1$
	\begin{equation}\label{prop:APSS.est} \|\chiR R\e^{-\i tT}g_E(T)\chir r\| \le C_{n,\eps} \left(\frac{\sfc_{E}+\eps}{R-r}\right)^n \qquad \mbox{for all}\ R>r>0\ \mbox{and}\ t\in\left[0,\frac{R-r}{\sfc_E+\eps}\right]. \end{equation}
\end{prop}
\begin{rem}\ 
	\begin{enumrem}
		\item
		Self-adjointness (and lower-semiboundedness) of $T$ is immediate by the Kato--Rellich theorem. Hence, the unitary group $\e^{-\i t T}$ and $g_E(T)$ are well-defined, by spectral calculus.
		Further, it follows that
		\begin{equation}\label{eq:cE}
			\sfc_E \le C_V (1\vee E)^{1/2}
		\end{equation}
		for some solely $V$-dependent constant $C_V>0$.
		\item 
		The function $g_E$ smoothly decays from one to zero on an interval of the length $(\alpha-1) E$, see \cref{fig:gE}.
		The linear dependence on $E$ is essential to the uniformity of the constant $C_{n,\eps}$ in $E$, cf. \cref{corgE}.
		\item A similar estimate to \cref{prop:APSS.est} can be found in \cite[Eq.~(2.5)]{ArbunichPusateriSigalSoffer.2021} and is proven in \textsection\,4 of that article for $g\in\cC^\infty_0(\IR)$. However, the dependence of the upper bounds on $\sup \supp g$ is not apparent in that article, whence we follow the proof here and derive a bound uniform in $E$. We emphasize that our explicit constants can also be used to derive the dependence on $g$ for more general classes than our choice \cref{def:gE}.
	\end{enumrem}
\end{rem}
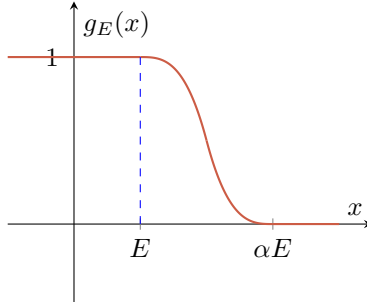
\begin{figure}[hbt]
	\caption{The function $g_E$ preserves energies up to $E$ and vanishes for values greater than $\alpha E$, $\alpha > 1$,  with a smooth interpolation in between.}\label{fig:gE}
	\centering
	\begin{tikzpicture}[declare function={xmi=-2;xma=8;a=2;b=6;c=1.5;}]
		\begin{axis}[xmin=xmi, xmax=xma+1, ymin=-0.7, ymax=2,
			xtick = {0}, ytick = {0},
			xlabel={$x$}, ylabel={$g_E(x)$},
			extra x ticks={a,b},
			extra x tick labels={$E$,$\alpha  E$},
			extra y ticks={c},
			extra y tick labels={$1$},
			scale=0.7,
			axis x line=center, axis y line= center,
			samples=40]
			\addplot[darkcoral, samples=50, domain=xmi:a, thick]
			plot (\x, c);
			\addplot[darkcoral, samples=50, domain=b:xma, thick]
			plot (\x, 0);   
			\addplot[darkcoral,samples=300,domain=a:{(a+b)/2},thick] {-4*c*((x-a)/(b-a))^3+c};
			\addplot[darkcoral,samples=300,domain={(a+b)/2}:b,thick] {4*c*((b-x)/(b-a))^3};
			\addplot +[mark=none,dashed] coordinates {(a, 0) (a, c)};
		\end{axis}
	\end{tikzpicture}
\end{figure} 
\begin{proof}
	We for now fix $\eta\in(0,R-r)$ and write
	\begin{equation}\label{def:bx}
		\bx as \coloneqq s^{-1}(\braket{x}_{\!\eta}-\eta-a) \quad \mbox{with}\quad \braket{x}_{\!\eta}\coloneqq\sqrt{\eta^2+\abs{x}^2} \quad  \mbox{for}\ x\in\IR^d,\ a>0.
	\end{equation}
	In the following, we also write  $x$ and $\bx as$ for the corresponding multiplication operators.
	The main ingredient to the proof is an estimate for the commutator $\big[g_E(T),\chi(\bx as)\big]$ for $a>0,s\geq 1$, where $\chi\in\cC_\sfb^{\infty}(\IR)$ with $\chi'\in\cC_0^\infty(\IR)$.
	Explicitly, employing methods from \cite{IvriiSigal.1993,HunzikerSigal.2000,ArbunichPusateriSigalSoffer.2021}, we use an expansion of the form
	\begin{equation}\label{eq:commexpgE}
		\big[g_E(T),\chi(\bx as)\big]
		=
		\sum_{k=1}^{n-1}\frac{s^{-k}}{k!}B_{E,k}\chi^{(k)}(\bx as) + s^{-n}R_{E,\chi,n}(a,s),
	\end{equation}
	with bounded operators $B_{E,k}$ and $R_{E,\chi,n}(a,s)$, $E,a>0,s \geq 1$, $k, n\in \IN$. They satisfy
	\begin{equation}\label{eq:Euniformbounds}
		\beta_{n} \coloneqq \max_{k=1,\ldots,n}\max_{j=0,1}\sup_{E\ge 1} \norm{\wh p^j \! B_{E,k}}< \infty
		\quad \mbox{and}\quad
		\rho_{\chi,n} \coloneqq \max_{j=0,1}\sup_{E,s\ge 1,a>0} \norm{\wh{p}^j \! R_{E,\chi,n}(a,s)} < \infty,
	\end{equation}
	see \cref{lem:commexp,corgE} for details (where we use the shorthand notation $R_{E,\chi,n}$ for $R_{g_E,\chi,n}$).
	
	For now, assume $0\le u\in \cC_0^\infty(\IR)$ and $v>\sfc_E$. Further, we define
	\begin{equation}\label{def:Phi}\Phi_{t,s} \coloneqq \int_{-\infty}^{\bx{r+vt}{s}}u(\lambda)^2\d\lambda  \qquad \mbox{as well as} \qquad \sfD\Phi_{t,s}\coloneqq \i[T,\Phi_{t,s}]+\partial_t\Phi_{t,s}\end{equation} and follow the lines of \cite[pp.\,7--8]{ArbunichPusateriSigalSoffer.2021}.
	Explicit calculation of the commutator $[T,\Phi_{t,s}]$ gives
	\begin{equation}\label{eq:DPhi} \sfD \Phi_{t,s} = \frac 1su (\bx{r+vt}{s})(\gamma-v)u(\bx{r+vt}{s}), \qquad \mbox{with}\ \gamma \coloneqq \frac{1}{2} \left( \wh p \cdot (\nabla\!\braket{x}_{\!\eta}) + (\nabla\!\braket{x}_{\!\eta}) \cdot \wh p \right).  \end{equation}
	We want to derive an upper bound for $g_E(T)\sfD \Phi_{t,s} g_E(T)$. To this end, we first estimate with Cauchy-Schwarz
	\begin{equation}\label{eq:APSS41}
		\abs{\braket{\psi,g_E(T) u(\bx{r+vt}{s}) \gamma u(\bx{r+vt}{s}) g_E(T) \psi}}
		\le \norm{(\nabla\!\braket{x}_{\!\eta}) u(\bx{r+vt}{s}) g_E(T) \psi} \norm{\wh p u(\bx{r+vt}{s}) g_E(T) \psi}.
	\end{equation}
	Using $g_{\alpha E}  g_E = g_E$ and applying \cref{eq:commexpgE} with $E$ replaced by $\alpha E$ and $\chi=u$, we have
	\begin{equation}\label{eq:commgt}
		\begin{aligned} 
			& \wh p u(\bx{r+vt}{s}) g_E(T)  =  \wh p\, g_{\alpha E}(T) u(\bx{r+vt}{s})g_E(T) + \wh p \big[ u(\bx{r+vt}{s}) , g_{\alpha E}(T) \big] g_E(T) \\
			&= \wh p\, g_{\alpha E}(T) u(\bx{r+vt}{s}) g_E(T)
			+ \sum_{k=1}^{n-1} \frac{s^{-k}}{k!} \wh p  B_{\alpha E,k} u^{(k)}(\bx{r+vt}{s}) g_E(T)
			+ s^{-n} \wh p R_{\alpha E,u,n}(r+vt,s) g_E(T) \nn \psi^2
			. \end{aligned}
	\end{equation}
	Using \cref{eq:Euniformbounds} as well as $\norm{\wh p\, g_{\alpha E}(T)}\le \sfc_E$ in above inequality
	to estimate the second term in \cref{eq:APSS41} and $\abs{\nabla \!\braket{x}_{\!\eta}}\le 1$ to estimate the first term, we arrive at
	\begin{equation}\label{eq:APSS42}
		\begin{aligned}
			&\abs{\braket{\psi,g_E(T) u(\bx{r+vt}{s}) \gamma u(\bx{r+vt}{s}) g_E(T) \psi}}
			\\& \qquad \qquad 
			\le  \sfc_E\norm{u(\bx{r+vt}{s}) g_E(T) \psi}^2
			+ \frac{ \beta_{n} }{ s}  \norm{\wt u(\bx{r+vt}{s}) g_E(T) \psi }^2
			+ \frac{\rho_{u,n}}{s^n} \nn \psi^2  ,
		\end{aligned}
	\end{equation}	
	where we defined
	\begin{equation}\label{def:ut}  \wt u \coloneqq \left( u^2 + \sum_{k=1}^{n-1} \frac{(u^{(k)})^2}{(k!)^2s^{2k-1}}  \right)^{1/2}  \in \cC_0^\infty(\IR) \quad \mbox{and for later reference note}\quad \supp \wt u = \supp u. \end{equation}
	Reinserting \cref{eq:APSS42} into \cref{eq:DPhi} yields
	\begin{equation}
		\label{eq:geT D geT}
		g_E(T)\sfD \Phi_{t,s} g_E(T)
		\le \frac{ \sfc_E-v}{s}g_E(T) u^2(\bx{r+vt}{s}) g_E(T)
		+ \frac{ \beta_{n} }{ s^2} g_E(T)\wt u ^2 (\bx{r+vt}s) g_E(T)
		+ \frac{\rho_{u,n}}{s^{n+1}}  \nn \psi^2.
	\end{equation}
	Given any $\phi\in L^2(\IR^d)$, we can take the matrix element in above inequality w.r.t. $\e^{-\i T t}\chir r\phi$ and set $\psi_t \coloneqq \e^{-\i T t}g_E(T)\chir r\phi$ for $t\ge 0$.
	Observe that $\sfD\Phi_{r,s}$ in \eqref{def:Phi} is defined in such a way that
	\[\Braket{\psi_t,\Phi_{t,s} \psi_t} = \Braket{\psi_0,\Phi_{0,s}\psi_0} + \int_0^t \Braket{\psi_r,\sfD\Phi_{r,s}\psi_r} \d r.\]
	Combining this with \eqref{eq:geT D geT} yields
	\begin{equation}\label{eq:ftmainbound}
		\begin{aligned}
			\Braket{\psi_t,\Phi_{t,s}\psi_t}
			&+ \frac{v-\sfc_E}{s}\int_0^t \Norm{u^2(\bx{r+vt}s)\psi_r}^2\d r 
			\\& \qquad 
			\le \braket{\psi_0,\Phi_{0,s}\psi_0}
			+ \frac{ \beta_{n} }{ s^2}   \int_0^t \Norm{\wt u^2(\bx{r+vt}{s})\psi_r}^2\d r
			+ \frac{ \rho_{u,n}}{s^{n}} \frac ts \Norm{\phi}^2.
		\end{aligned}
	\end{equation}
	This is the main estimate in the proof of the statement, cf. \cite[Eq.~(4.13)]{ArbunichPusateriSigalSoffer.2021}.
	
	Now, following \cite[p.\,9]{ArbunichPusateriSigalSoffer.2021}, we want to approximate $\chiR {R}$ by a smooth function.
	To this end, for $\tau>0$,
	let 
	\begin{equation}
		\cX_{\tau} \coloneqq \left\{ \chi\in\cC^\infty(\IR;[0,1]) \colon \supp \chi\subset (0,\infty),\ \chi=1\ \mbox{on}\ [\tau,\infty),\ \chi'\ge0,\ \sqrt{\chi'}\in\cC^\infty(\IR)  \right\}
	\end{equation}
	and assume $\chi\in \cX_{\tau}$. If $\bx {r}s \in \supp \chi$, it follows that $\braket{x}_{\!\eta}>r + \eta$, so 
	$\supp \chi(\bx rs)\cap \supp \chir r = \emptyset$.
	Applying the commutator expansion \cref{eq:commexpgE}, we hence find
	\begin{equation}\label{eq:f0bound}
		\chi(\bx {r}s)g_E(T)\chir r = \big[ \chi(\bx {r}s),g_E(T) \big] \chir r = - s^{-m}R_{E,\chi,m}(r,s) \chir r \quad \mbox{for all}\ m\in\IN.
	\end{equation}
	Inserting \cref{eq:f0bound,eq:Euniformbounds} into \cref{eq:ftmainbound} with $\Phi_{t,s} = \chi(\bx{r+vt}{s})$, i.e., we set $u:=\sqrt{\chi'}$ in \cref{def:Phi}, and neglecting the first term on the left hand side, we obtain
	\begin{align*}
		\int_0^t \Braket{\psi_r,\chi'(\bx{r+vt}s)\psi_r}\d r
		\le\ &
		\frac 1s \frac{\beta_{m}}{v-\sfc_{E}}
		\int_0^t \Norm{\wt u^2(\bx{r+vt}{s})\psi_r}^2\d r
		+ \frac{1}{s^{m-1}}\left( \frac{\rho_{\chi,m} + \rho_{u,m}ts^{-1}}{v-\sfc_{E}} \Norm{\phi}^2
		\right).
	\end{align*}
	Observing that, by its definition \cref{def:ut}, $\wt u = \sqrt{\wt \chi'}$ for some $\wt \chi \in \cX_\tau$, we can iterate this bound to obtain a sequence $(\chi_k)_{k\in\IN_0}\subset \cX_\tau$ (recursively defined by $\chi_{k+1}\coloneqq \wt{\chi_k}$, $\chi_0=\chi$ and setting $u_k\coloneqq \sqrt{\chi_k'}$) such that for $k=0,\ldots,m-1$,
	\begin{align*}
		\int_0^t & \Braket{\psi_r,\chi'(\bx{r+vt}{s})\psi_r} \d r
		\\ & 
		\le
		\frac{1}{s^{k+1}}\frac{\prod_{j=0}^{k} \beta_{m-j}}{(v-\sfc_E)^{k+1}}\int_0^t\Braket{\psi_r,\chi_{k+1}'(\bx{r+vt}{s})\psi_r} \d r
		+ \frac{\norm\phi^2}{s^{m-1}}\sum_{\ell=0}^k\frac{ \left( \prod_{j=0}^{\ell-1} \beta_{m-j} \right) (\rho_{\chi_\ell,m-\ell} + \rho_{u_\ell,m-\ell}ts^{-1})}{(v-\sfc_E)^\ell}.
	\end{align*}
	Setting $k = m-1$, we obtain for $0 < t \leq s$,
	\begin{align}
		\label{eq:derintsn}
		\int_0^t & \Braket{\psi_r,\chi'(\bx{r+vt}{s})\psi_r} \d r  \le 
		\frac{\norm\phi^2}{s^{m-1}}\underbracket{\left(\frac{\beta_m\cdots \beta_{1}}{(v-\sfc_E)^m}\norm{\sqrt{\chi_m'}}_\infty^2 + \sum_{\ell=0}^{m-1}\frac{\left( \prod_{j=0}^{\ell-1} \beta_{m-j} \right) (\rho_{\chi_\ell,m-\ell} + \rho_{u_\ell,m-\ell})}{(v-\sfc_E)^\ell} \right)}_{\eqqcolon \Omega_m(\chi,v)}.
	\end{align}
	Again inserting \cref{eq:f0bound,eq:derintsn} (with $m=n-1$) into \cref{eq:ftmainbound} and neglecting the second term on the left hand side (which is possible by $v>\sfc_E$), we arrive at 
	\begin{equation}\label{eq:finalAPSS}
		\Braket{\psi_t,\chi(\bx{r+vt}{s})\psi_t}
		\le
		\frac{\|\phi\|^2}{s^n}\left( \rho_{\chi,n} + \rho_{\sqrt{\chi'},n} +  \beta_{n}\Omega_{m-1}(\wt\chi,v) \right).
	\end{equation}
	We now fix the previously introduced parameters as follows:
	\[v = \sfc_E+\frac\eps2, \quad \tau=\frac \eps 2 \quad \mbox{and}\quad s=1 \vee \frac{R-r-\eta}{\sfc_E+\eps}.\]
	Then
	\[ \bx{r+vt}s = \sfc_E + \eps -(\sfc_E + \eps/2)\frac ts  \ge \eps/2=\tau
	\qquad \mbox{if}\ \abs x \ge R\ \mbox{and}\ 0<t\le s = 1 \vee  \frac{R-r-\eta}{\sfc_E+\eps}, \]
	so $\chiR R(x) \le \chi({\bx{r+vt}{s}})^{1/2}$ for any $\chi\in\cX_{\eps/2}$.
	Combined with \cref{eq:finalAPSS}, we arrive at
	\[ \|\chiR R\e^{-\i t T}g_E(T)\chir r \phi \| \le \Braket{\psi_t,\chi(\bx{r+vt}{s})\psi_t}^{1/2} \le \widetilde C_{n,\eps} \nn \phi \left(\frac{\sfc_{E}+\eps}{R-r-\eta}\right)^{n/2} \quad \mbox{for}\ t\in\left[0,s\right], 
	\]
	where the constant $\widetilde C_{n,\eps}$ is given by the maximum of one and the square root of the bracket on the right hand side of \eqref{eq:finalAPSS}. 
	Observing that $\widetilde C_{n,\eps}$ is independent of $\eta\in(0,R-r)$, we can take the limit $\eta\downarrow 0$ and set $C_{n,\eps} := \widetilde C_{2n,\eps}$, which finishes the proof.
\end{proof}
To apply above \lcnamecref{prop:APSS}, we will need the fact that a Schwartz-function remains spatially and energetically localized after the application of $T$.
\begin{lem}
	\label{th:tail estimates}
	Assume \cref{hyp:Lbound} holds.
	\begin{enumlem}
		\item\label{th:tail.pos}
		For all $\ph\in\cS$, $n \in \IN$ and  $j\in\{0,1\}$, there exists $\Ctailx{n,j}\ge0$ such that
		\[
		\Norm{\chiRx{R}{x} T^j \ph_x} \leq \Ctailx{n,j}\norm{\chr_{B_R^\sfc}V}_\infty^j R^{-n} \qquad \mbox{for all}\ x\in\IR^d,\ R>0. 
		\] 
		If $\ph=\ph_\sigma$, we have $\Ctailx{n,j} \lesssim \sigma^{n+\frac{1-d}2-2j}$.
		\item\label{th:tail.en}
		Assume \cref{hyp:Vdiff} holds and let $n\in\IN$, $j\in\{0,1\}$ such that $2(n+j)\le n_V$. Then,
		for all $\ph\in\cS$,
		there exists a constant $\CtailE{n,j}$ such that
		\[
		\nn{\chr_{[E,\infty)}(T) T^j \ph_x} \leq  \frac{\CtailE j}{E^n} \qquad \mbox{for all}\ x\in\IR^d,\ E>0.
		\]
		If $\ph=\ph^\sigma$, we have $\CtailE{n,j}  \lesssim  \sigma^{-d/2} (1 \vee \sigma^{-2(n+j)})$.
	\end{enumlem}
\end{lem}
\begin{proof}[Proof of \subcref{th:tail.pos}]
	First, we observe that the
	statement is trivial for a generic constant and arbitrary Schwartz functions $\ph$,
	by the observation
	\begin{align}\label{eq:decaysimple}
		\Norm{\chiRx{R}{x} T^j \ph_x} \le \Norm{\chiRx{R}{x} (-\Delta)^j \ph_x} + \norm{\chr_{B_R^\sfc}V}_\infty^j\Norm{\chiRx{R}{x} \ph_x}.
	\end{align}
	It remains to treat the case of Gaussians.
	Therefore, we apply the usual Gaussian tail estimate
	\begin{align}
		\int_{r \geq c}\e^{-a^2r^2} \d r &\leq \int_{ r\geq c} \frac rc \e^{-a^2r^2}\d r  = \frac{1}{2ca^2}\e^{-c^2a^2} , \label{eq:tail1}
	\end{align}
	as well as
	\begin{align}\label{eq:gaussianpoly}
		r^{m} \leq C_m e^{r^2 /2} \qquad \mbox{for all}\ m\in\IN,
	\end{align}
	where $C_m>0$ solely depends on $m$. This yields
	\begin{align*}
		\Norm{\chiRx{R}{x}\ph_x^\sigma}^2 &= (\pi\sigma^2)^{-d}\int_{\abs{y-x}\ge R} \e^{-|y-x|^2/\sigma^2}\d y
		\\& = (\pi\sigma^2)^{-d} \sigma^{d-1} A_d \int_{r \geq R} (r/\sigma)^{d-1} \e^{-r^2/\sigma^2} \d r  \\
		& \leq \frac{C_{d-1} A_d}{  \sigma^{d+1} \pi^{d} } \int_{r \geq R}  \e^{-r^2/(2\sigma^2)} \d r 
		\leq \frac{2 C_d A_d}{ \sigma^{d-1} \pi^{d} } \e^{-R^2 / (2 \sigma^2)}
		\le \frac{2C_{d-1} A_d C_{n}}{\pi^{d}}\cdot \frac{\sigma^{n+1-d}}{R^n}
	\end{align*}
	where $A_d$ denotes the volume of the $d-1$-dimensional unit sphere. 
	For $j=1$, we again use \cref{eq:tail1,eq:gaussianpoly} and obtain
	\begin{align*}
		\Norm{\chiRx{R}{x}  (-\Delta) \ph_x^\sigma}^2 &= C_d (\sigma^2)^{-d}\int_{r \ge R} r^{d-1} \left(   \frac{r^2}{\sigma^4} - \frac{d}{\sigma^2}  \right)^2 \e^{-r^2/\sigma^2}\d r \\
		&\leq C_d { (\sigma^2)^{-d}} \sigma^{d-1} \frac{1}{\sigma^4}  \int_{r \ge R} \e^{-r^2/ 2\sigma^2}\d r \\
		&\leq C_d (\sigma^2)^{-d} \sigma^{d-3} \frac{1}{R}   \e^{-R^2/ 2\sigma^2} .
	\end{align*}
	Inserting these bounds into \cref{eq:decaysimple} proves the statement.
\end{proof}
\begin{proof}[Proof of \subcref{th:tail.en}]
	Notice that by assumption $\ph_x \in \cD(H^{n})$, thus by functional calculus, for any $m\in\IN$,
	\begin{align*}
		\nn{\chr_{[E,\infty)}(T) T^j \ph_x}^2 &= \sc{\ph_x, \chr_{[E,\infty)}(T) T^{2j}\ph_x} \leq \nn{ \frac{T^{m+j}}{E^{m}} \chr_{[E,\infty)}(T)  \ph_x}^2
		\leq \frac{1}{E^{2m}} \nn{T^{m+j} \ph_x}^2 .
	\end{align*}
	The statement now follows, by observing that \cref{hyp:Vdiff} implies there exists a constant $C_{V}$, solely depending on $V$, such that
	\begin{align}\label{eq:Tjbound}
		\sup_{x\in\IR^d}\norm{T^m\ph_x}
		\le
		C_{V}\max_{k=0,\ldots,2m} \max_{\substack{\alpha\in\IN^d\\\norm\alpha_1=k}}\norm{\partial_\alpha\ph}
		\qquad\mbox{for all}\ m\in\IN,\ 2m\le n_V.
	\end{align}
	For the Gaussian case we observe that each derivative applied to the Gaussian will produce an additional factor of $\sigma^{-1}$, i.e.,
	\begin{align}\label{eq:Gaussiandiff}
		\nn{\partial_\alpha \ph} \lesssim \sigma^{-d/2}( 1 \vee \sigma^{-2\norm\alpha_1} ). 
	\end{align}
	This completes the proof.
\end{proof}
From the propagation velocity bound, we now directly obtain the following Lieb--Robinson type bound. For the use therein and from now on, we introduce the polynomially decaying functions
\begin{equation}\label{def:Gnt}
	G_{n,t}(r) \coloneqq 1\wedge\frac{\braket t^n}{r^n}, \quad  G_{n,t}(x) \coloneqq G_{n,t}(\abs x), \qquad n \in \IN,~ r \geq 0, ~ x \in \IR^d.
\end{equation}
\begin{prop}\label{lem:Rfscalar}
	Let $\delta  > 0$, $n \in \IN$ and $\ph\in\cS$.
	\begin{enumlem}
		\item Assume $2n \leq n_V$. Then there exists a constant $\Cobt{0}  > 0$  such that
		\[
		\Abs{\braket{f,\e^{-\i t T}\ph_x}} \le  \Cobt{0} \nn f \braket{t}^\delta 
		G_{n,t}(R_f)
		\]
		for all $x \in \IR^d$, $t\in\IR$ and $f\in L^2(\IR^d)$ with $R_f\coloneqq\inf\abs{\supp f - x}\ge 1$.
		
		\item 
		Assume $2n+2 \leq n_V$. Then there exists a constant $\Cobt{1}  > 0$  such that
		\[
		\Abs{\braket{f,(\e^{-\i t T}-\Id)\ph_x}} \leq  \Cobt{1} \nn f \abs t \braket{t}^\delta  G_{n,t}(R_f)
		\]
		for all $x \in \IR^d$, $t\in\IR$ and $f\in L^2(\IR^d)$ with $R_f\coloneqq\inf\abs{\supp f - x}\ge 1$.
	\end{enumlem}
	Moreover, in the case $\ph=\ph_\sigma$, we have for $\sigma \leq 1$,
	\begin{align*}
		\Cobt{j} \lesssim  \sigma^{-d/2} \sigma^{-2(n(n+\delta)/2\delta+j)}   \qquad\mbox{for}\ j\in\{0,1\}.
	\end{align*}
\end{prop}
\begin{proof}
	For now assume that $B$ is a bounded operator on $L^2(\IR^d)$ and $j\in\{0,1\}$.
	From the simple decomposition
	\begin{align*}
		\Braket{f,\e^{-\i t T} T^j\ph_x} = & \Braket{\chir{ R_f} f,\e^{-\i t T}T^j\ph_x}
		\\&
		+ \Braket{\chiR{ R_f} f,\e^{-\i t T} (1-B) T^j\ph_x} + \Braket{\chiR{ R_f} f, \chiR{ R_f} \e^{-\i t T} B T^j\ph_x},
	\end{align*}
	the assumption $\chir{ R_f} f=0$, the trivial fact $\nn{\e^{-\i t T}} =1 $, and the Cauchy--Schwarz inequality, we obtain
	\begin{equation}\label{eq:decomp}
		\abs{ \sc{ f, e^{-\i t T} T^j \ph_x }} \leq  \nn{ f} \nn{ e^{-\i t T} (1-B) T^j \ph_x}  + \nn{ \chiRx R x e^{-\i t T} B} \nn{  f} \nn{T^j \ph_x}.
	\end{equation}
	For some $E\ge 1$,
	we now set $B=g_E(T)\chirx{R_f/2}{x}$, with $g_E$ as defined in \cref{prop:APSS},
	and apply the identity $1-B = (1-g_E(T)) + g_E(T)\chiRx{R_f/2}{x}$ to estimate the first term.
	For $m\in\IN$ with $2(m+j)\le n_V$, \cref{th:tail estimates} implies for 
	\begin{align}
		\nn{e^{-\i t T} (1-B) T^j \ph_x} &\leq \nn{(1-g_E(T)) T^j\ph_x} + \nn{\chiRx{R_f}{x} T^j \ph_x} \nonumber \\ 
		&\leq \nn{ \chr_{[\alpha E,\infty)}(T) T^j \ph_x} + \nn{\chiRx{R_f}{x}  T^j \ph_x}  \nonumber \\ 
		&\leq   \CtailE{m,j} E^{-m}    + \Ctailx{n,j}    R_f^{-n} \label{eq:gEphbound}
	\end{align}
	To estimate the second term in \eqref{eq:decomp} we use \cref{prop:APSS,eq:cE}.  For any $\delta > 0$, this yields a constant $C_\delta > 0$ such that
	\begin{equation}\label{eq:APSSappl}
		\nn{ \chiRx{R_f} x e^{-\i t T} B} \le C_{\delta} \left(\braket{t}\sqrt E/R_f\right)^{n+ \delta}.
	\end{equation}
	We now choose $E=R_f^{2\delta/(n+\delta)}$ and $m=n(n+\delta)/2\delta$, whence
	\begin{align}
		\label{eq:m chosen}
		E^{-m} = R_f^{-n} \quad \mbox{and} \quad \left(\braket{t}\sqrt E/R_f\right)^{n+\delta} = \braket{t}^\delta \big(\braket{t}/R_f\big)^n. 
	\end{align}
	Inserting \cref{eq:m chosen,,eq:APSSappl,,eq:gEphbound} into \eqref{eq:decomp}, we find
	\begin{align*}
		&\Abs{\braket{f,\e^{-\i t T} T^j\ph_x}}  \leq   \frac{\nn f}{R_f^n}  ( \Ctailx{n,j}     + \CtailE{m,j}  +  C_\delta \braket{t}^\delta \braket t^n \nn{T^j \ph_x} ).
	\end{align*}
	This directly proves the first statement by considering $j=0$ and, in particular, the $\sigma$-dependence in the case of Gaussians follows from the $\sigma$-dependence of the tail bound constants as given in \cref{th:tail estimates,eq:Gaussiandiff,eq:Tjbound}. 
	The second statement similarly follows, by observing
	\[
	\frac{\sfd}{\sfd t} \braket{f,\e^{-\i t T}\ph_x} = \braket{f,\e^{-\i t T} (-\i T)\ph_x}
	\]
	and the above calculation with $j=1$. 
\end{proof}
By an appropriate dyadic expansion of an arbitrary $L^2$-function $f$, we now obtain our desired one-body Lieb--Robinson bound.
\begin{proof}[\textbf{Proof of \cref{prop:LRint}}]
	Let $j \in\{0,1\}$, $f_0(y) =\chr_{[0,1]}(\abs{y-x})  f(y)$ and $f_k(y) = \chr_{[2^{k-1},2^{k}]}(\abs{y-x}) f(y)$ for $k\in\IN$. Using the notation from \cref{lem:Rfscalar} this immediately implies $R_{f_k}=2^{k-1}$ and hence
	\begin{align*}
		\abs{\braket{f_k,(\e^{-\i t T} - j \Id)\ph_x}} \le \Cobt{j} \abs t^j \braket{t}^\delta  G_{n,t}(2^{k-1})\nn{f_k} \quad \mbox{for}\ k\in\IN.
	\end{align*}
	Summing over $k\in\IN_0$ and applying the Cauchy-Schwarz inequality, we find
	\begin{align*}
		&\abs{\braket{f,(\e^{-\i t T} - j \Id) \ph_x}} \leq
		\sum_{k=0}^\infty \abs{ \sc{ f_k, (\e^{-\i t T} - j \Id) \ph_x}} \leq  2\nn{f_0} + \Cobt{j} \abs t^j \braket{t}^\delta\sum_{k=1}^\infty  G_{n,t}(2^{k-1})\nn{f_k} 
		\\
		& \leq  (2\vee \Cobt{j}) \abs t^j \braket{t}^\delta\left(1+\sum_{k=1}^\infty  G_{n,t}(2^{k-1}) \right)^{1/2} \left(\nn{f_0}^2 + \sum_{k=1}^\infty  G_{n,t}(2^{k-1}) \nn{f_k}^2 \right)^{1/2}
	\end{align*}
	We can easily estimate the first factor on the right hand side, using the monotonicity of $ G_{n,t}$, by
	\begin{align*}
		1+\sum_{k=1}^\infty  G_{n,t}(2^{k-1}) \le 1+ \int_0^\infty  G_{n,t}(r)\d r = 1+ \int_0^{\braket{t}} 1 \d r + \int_{\braket t}^\infty\left(r/\braket t\right)^{-n} \d r \le \left(2+\frac{1}{n-1}\right)\braket t.
	\end{align*}
	Further, we estimate the second factor using
	\begin{align*}
		\nn{f_0}^2 + &\sum_{k=1}^\infty  G_{n,t}(2^{k-1}) \nn{f_k}^2 
		= \int_{\abs{\cdot-x}\le 1} \abs{f(y)}^2 \d y + \sum_{k=1}^\infty \int_{2^{k-1}\le \abs{\cdot-x}\le 2^k}  G_{n,t} (2^{k-1})\abs{f(y)}^2 \d y\\
		& \le \int_{\abs{\cdot-x}\le 1}  G_{n,t}(y-x) \abs{f(y)}^2 \d y + \sum_{k=1}^\infty \int_{2^{k-1}\le \abs{\cdot-x}\le 2^k}  G_{n,t} ((y-x)/2)\abs{f(y)}^2 \d y\\
		&\leq   2\int_{\IR^d}  G_{n,t}(y-x) \abs{f(y)}^2 \d y.
	\end{align*}
	Combining the above estimates yields the statement. Especially, the $\sigma$-dependence in the case $\ph=\ph^\sigma$ is obvious from \cref{lem:Rfscalar}.
\end{proof}


\section{From One-Particle to Many-Body Lieb--Robinson Bound}\label{sec:manybody}

In this \lcnamecref{sec:manybody}, we follow the lines of \cite{GebertNachtergaeleReschkeSims.2020} to prove \cref{thm:manybody}, i.e., obtain the Lieb--Robinson bound for interacting fermions.
Throughout, we will assume \cref{hyp:V,hyp:W} hold and fix some $\ph\in\cS$.
\begin{lem}\label{lem:GNRSest}
	Let
	\begin{equation}\label{def:kernel}
		K_t(f,x) \coloneqq \norm{ W }_1 \Abs{\braket{\e^{-\i t T}f,\ph_x}} +  \abs W \ast \Abs{\braket{\e^{-\i t T }f,\ph_\bullet}}(x).
	\end{equation}
	Then
	\begin{equation}\label{eq:iteration-squared}
		F_t^\Lambda(f,g)^2 \le 2t \nn{\ph}^4 \int_0^t K_{t-s}(f,x)^2\left(\Abs{\braket{\e^{-\i s T} \ph_x,g}}^2 + F_s^\Lambda(\ph_x,g)^2\right)\d s  \qquad \mbox{for all}\ t>0.
	\end{equation}
\end{lem}
\begin{proof}
	From \cite[Lemma~4.1]{GebertNachtergaeleReschkeSims.2020}, we know that
	\begin{equation}\label{eq:iteration-GNRS}
		F_t^\Lambda(f,g) \le \nn{\ph}^2 \int_0^t K_{t-s}(f,x)\left(\Abs{\braket{\e^{-\i s T} \ph_x,g}} + F_s^\Lambda(\ph_x,g)\right) \d s.
	\end{equation}
	We emphasize that the derivation is independent of the concrete choice of $\ph$. Squaring \cref{eq:iteration-GNRS} and afterwards applying the Cauchy--Schwarz inequality  yields the statement.
\end{proof}
We now estimate the integral kernel, applying the one-body Lieb--Robinson bound \cref{prop:LRint}.
\begin{lem}\label{lem:kernelest}
	Let the assumptions of \cref{prop:LRint} be satisfied.
	For all $n\leq \frac{n_V}{2} \wedge n_W$, there exists a solely $n$-dependent constant $C_n$ such that
	\begin{align}
		\label{eq:kernelest}
		K_t(f,x)^2 \leq 2 \nn{\ph}^4( \norm{W}_1^2 +   \cw \nn W_1 C_n   )  \Cob{0}   \sc t^{1+2\delta + d}\int  G_{n,t}(x-y) \abs {f(y)}^2\d y
	\end{align}
	for $\delta > 0$, $f\in L^2(\IR^d)$, $x\in\IR^d$ and $t\in \IR$.
\end{lem}
\begin{proof}
	Without loss of generality assume that $\nn{\ph} = 1$. 
	Starting from \eqref{def:kernel},  we first use the trivial estimate
	\begin{align}
		\label{eq:first trivial}
		K_t(f,x)^2 \le 2\norm{W}_1^2 \abs{\braket{\e^{-\i t T}f,\ph_x}}^2 + 2 \left(\abs W \ast \abs{\braket{\e^{-\i t T}f,\ph_\bullet} } (x) \right) ^2.
	\end{align}
	We can directly bound the first term using \cref{prop:LRint} and obtain the first summand on the right-hand side of \eqref{eq:kernelest}. 
	Further, applying the Cauchy-Schwarz inequality, we obtain
	\begin{align*}
		\left( \abs{W} * \abs{\sc{e^{-\i t T} f, \ph_\bullet }} (x) \right)^2 &=  \left(\int \abs{W(x-y)} \abs{\sc{e^{-\i t T} f, \ph_y }} \d y \right)^2 \\
		&\leq   \int \abs{W(x-y)} \d y  \int \abs{W(x-y)} \abs{\sc{e^{-\i t T} f, \ph_y }}^2 \d y .
	\end{align*}
	Now, we apply \cref{prop:LRint} in the second factor and subsequently use \cref{lem:simconv}, which yields \[\abs W \ast G_{n,t} (x)  \leq  C_n c_W \braket t^d G_{n,t}(x)\] for some solely $n$-dependent constant $C_n$, so
	\begin{align*}
		\int  \abs{W(x-y)} \abs{\sc{e^{-\i t T} f, \ph_y }}^2 \d y &\leq  \Cob{0}  \sc t^{1+2\delta} \int \abs{W(x-y)} \int G_{n,t}(y-z) \abs {f(z)}^2\d z \d y \\
		&=  \Cob{0}   \sc t^{1+2\delta}  \abs W * (G_{n,t}* \abs f^2) (x) 
		=   \Cob{0}  (\abs W * G_{n,t}) * \abs f^2 (x) 
		\\
		&\leq  \Cob{0}  \sc t^{1+2\delta + d}  \cw C_n G_{n,t} *  \abs f^2 (x) .
	\end{align*}
	Combining these estimates we get an upper bound for the second term on the right-hand side of \eqref{eq:first trivial} which yields the respective second term in \eqref{eq:kernelest}.
\end{proof}
Especially, this implies a pointwise bound for the case that $f$ itself decays polynomially.
\begin{cor}\label{cor:kernelGauss}
	Assume $n > d$, $\phi\in L^2 (\RR^d)$ satisfies $\abs{\phi(x)}^2\le C'_\ph G_{n,0}(x)$ for some constant $C'_\ph$. 
	Then there exists a constant $C_\ph \leq C'_\ph$ such that
	\[
	K_t(\phi_x,y)^2 \leq C_\phi \nn\ph^4 C_{n,W}  \Cob{0}  \braket t^{1+2\delta+2d}   G_{n,t}(x-y)
	\]
	for all $t \in \IR$ and $x,y \in \IR^d$, where $C_{n,W} \coloneqq 2( \norm{W}_1^2 +   \cw \nn W_1 C_n   ) $ is the prefactor from \cref{lem:kernelest}.
\end{cor}
\begin{proof}
	The statement follows directly from combining \cref{lem:kernelest,lem:simconv}. 
\end{proof}
\begin{rem}\label{rem:gaussdecay}
	If $\ph =\ph^\sigma$, we have $C_\ph \lesssim \sigma^{-d}$ and $\nn \ph^4 \lesssim \sigma^{-2d}$. 
\end{rem}
\begin{proof}[Proof of \cref{rem:gaussdecay}]
	In view of \cref{lem:kernelest} and the definition of $\ph^\sigma$ \eqref{def:gauss}, we have to estimate 
	\begin{align}
		\nonumber
		\int  G_{n,t}(y-z) \abs {\ph_x(z)}^2\d z &= \int G_{n,t}(y-z)  \frac{1}{(\pi\sigma^2)^d}  \exp\left(-(z-x)^2/\sigma^2\right) \d z 
		\\  &=   \frac{1}{(\pi\sigma)^d}  \int G_{n,t}(y-x-\sigma z)  \exp\left(-z^2 \right) \d z . \label{eq:integral to divide}
	\end{align}
	Now we find in the region $\sigma \abs z < \frac{\abs {y-x}}{2}$
	\begin{align}
		\label{eq:int div1}
		\int_{\sigma \abs z < \frac{\abs {y-x}}{2}}
		G_{n,t}(y-x-\sigma z)   \exp\left(-z^2 \right) \d z  \leq  G_{n,t}((y-x)/2) \int_{\IR^d} \exp\left(-z^2/2 \right) \d z.
	\end{align}
	On the other hand, we have for $\sigma \leq 1$
	\begin{align}
		\nonumber
		\int_{\sigma \abs z \geq \frac{\abs {y-x}}{2}}
		G_{n,t}(y-x-\sigma z)   \exp\left(-z^2 \right) \d z  
		&\leq 	\exp\left(- \frac{\abs {y-x}^2}{4\sigma^2} \right)  \int_{\IR^d}
		G_{n,t}(y-x-\sigma z)   \d z  \\
		\nonumber
		&\leq  	\exp\left(- \frac{\abs {y-x}^2}{8} \right)  	\exp\left(- \frac{\abs {y-x}^2}{8\sigma^2} \right) 	\int_{\IR^d} \	G_{n,t}(y-x-\sigma z)    \d z  \\
		&\leq \exp\left(- \frac{\abs {y-x}^2}{8} \right)  	\exp\left(- \frac{\abs {y-x}^2}{8\sigma^2} \right) 	\braket t^n \sigma^{-d} \int_{\IR^d} \frac{1}{\abs{z}^n \vee 1 }    \d z. 		\label{eq:int div2}
	\end{align}
	It is easy to see (e.g. with l'Hospital) that $\sup_{\abs x \geq 1,\sigma > 0} \exp\left(-x^2/(8\sigma^2) \right)	 \sigma^{-d}< \infty$ and the last integral is finite if $n > d$. The estimates \eqref{eq:int div1} and \eqref{eq:int div2} show that the last integral in \eqref{eq:integral to divide} can be estimated by $G_{n,t}(y-x)$ up to a $\sigma$-independent constant for all $\abs{y-x} \geq 1$. For $\abs{y-x} \leq 1$ this trivially holds as well. Thus, the behavior $C_\ph \lesssim \sigma^{-2d}$ follows from the prefactor in \eqref{eq:integral to divide} and because of the $L^1$ normalization of $\ph$, $\nn{\ph^\sigma} \lesssim \sigma^{-d/2}$.
\end{proof}
We now iterate the bound \cref{lem:GNRSest}. For any $N\in\IN$, this yields
\begin{align}
	\label{eq:iteration}
	F_t^\Lambda(f,g)^2  & \leq \sum_{k=1}^{N} S_k(t,f,g) + R_N(t,f,g), \qquad \mbox{where}\\
	\label{def:Sk}
	S_k(t_0,f_0,g) &\coloneqq \left(\prod_{\ell=1}^{k}\int_0^{t_{\ell-1}} \d t_{\ell} \int_{\IR^d}\d x_{\ell} 2t_{\ell-1} K_{t_{\ell-1}-t_{\ell}}(f_{\ell-1},x_{\ell})^2\right) \abs{\braket{\e^{-\i t_k T}\ph_{x_k},g}}^2,\quad f_\ell \coloneqq \ph_{x_{\ell}},\\
	\label{def:RN} 
	R_N(t_0,f_0,g) &\coloneqq \left(\prod_{\ell=1}^{N}\int_0^{t_{\ell-1}} \d t_{\ell} \int_{\IR^d}\d x_{\ell} 2t_{\ell-1}K_{t_{\ell-1}-t_{\ell}}(f_{\ell-1},x_{\ell})^2\right) F_{t_N}^\Lambda(\ph_{x_N},g)^2.
\end{align}
We want to apply \cref{lem:kernelest,cor:kernelGauss} to estimate above expressions.
Therein, we will apply the following simple bound.
\begin{lem}
	\label{lem:basic integral estimate}
	For all $k \in \IN_0$, $\alpha,\beta > 0$, we have
	\[
	\int_0^t \sc{t-s}^{\alpha} \sc s^{\beta} s^k \d s \leq \frac{\sc t^{\alpha+\beta}}{k+1} t^{k+1} .
	\]
\end{lem}
\begin{proof}
	Using the substitution $s' = s/t$, we obtain
	\begin{align*}
		\int_0^t \sc{t-s}^{1+2\delta} \sc s^{\beta} s^k \d s &= \int_0^t  (1+ (t-s)^2)^{\alpha/2} (1+ s^2)^{\beta/2} s^k \d s \\
		&= t^{k+1} \int_0^1  (1 +  t^2 (1-s)^2)^{\alpha/2}  (1 + t^2 s^2)^{\beta/2}   s^k \d s \\
		&\leq t^{k+1}  \int_0^1  (1 +  t^2)^{\alpha/2}   (1 + t^2)^{\beta/2}  s^k \d s \\
		&=  \frac{\sc t^{\alpha + \beta}}{k+1} t^{k+1}. \qedhere
	\end{align*}
\end{proof}
We can now derive the desired bounds.
\begin{lem}\label{lem:Skest}
	In the following statements, we use the constants defined in \cref{prop:LRint,cor:kernelGauss}.
	\begin{enumlem}
		\item
		\label{i:main term estimate}
		For all $f,g\in L^2(\IR^d)$, $t\in \IR$ and $k\in\IN$, we have
		\begin{align*}
			&S_k(t,f,g) \\&\leq (\Cob{0})^{k+1} C_{n,W}^{k}  ( C_\ph)^{k-1}  \nn \ph^{4k} 2^k \frac{\sc{t}^{(1+2\delta)(k+1) + d(3k-1)}  t^{2k-1} }{1 \cdot 3 \cdots (2k-1)} \int_{\IR^d} \abs{f(x)}^2 (G_{n,t} \ast \abs{g}^2)(x)  \d x.
		\end{align*}
		\item
		\label{i:remainder term estimate}
		For all $f,g\in L^2(\IR^d)$, $t\in\IR$ and $N\in\IN$, we have
		\begin{align*}
			&R_N(t,f,g) \\& \leq 36 \nn g^2 (C_{n,W} \Cob{0})^{N+1}  (C_\ph)^{N-1} \nn \ph^{4N} 2^N \frac{\sc{t}^{(1+\delta)(N+1) + d(2N-1)}  t^{2N-1} }{1 \cdot 3 \cdots (2N-1)} \int_{\IR^d} (G_{n,t}\ast \abs f^2) (x) \d x.
		\end{align*}
	\end{enumlem}
\end{lem}
\begin{proof}[Proof of \subcref{i:main term estimate}]
	Applying \cref{lem:kernelest,cor:kernelGauss,prop:LRint} to \eqref{def:Sk} we find
	\begin{align*}
		S_k(t,f,g)  \leq &(\Cob{0})^{k+1} C_{n,W}^{k} ( C_\ph)^{k-1}  \nn \ph ^2  
		\left( \int_0^{t} \d t_{1} \int_{\IR^d}\d x_1 2t \braket{t-t_1}^{1+2\delta+d}   G_{n,t-t_1} * \abs f^2(x_1) \right) 
		\\ &\times \left(\prod_{\ell=2}^{k} \int_0^{t_{\ell-1}} \d t_{\ell} \int_{\IR^d}\d x_{\ell} 2t_{\ell-1} \braket{t_{\ell-1}-t_{\ell}}^{1+2\delta+2d}   G_{n,t_{\ell-1}-t_{\ell}}(x_{\ell-1} - x_\ell)  \right) 
		\\ &\times \braket{t_k}^{1+2\delta}  G_{n,t} * \abs g^2(x_k)
	\end{align*}
	Now using that $t_k \leq t$ and $t_{\ell-1} - t_\ell \leq t$ for all $\ell$ and writing $t_0 := t$, we find
	\begin{align}
		\nonumber
		S_k(t,f,g)  & \leq (\Cob{0})^{k+1} C_{n,W}^k (C_\ph)^{k-1}  \nn \ph ^{4k}    2^k \left( \prod_{\ell=1}^k  t_{\ell-1} \int_0^{t_{\ell-1}} \d t_\ell \sc{ t_{\ell-1} - t_\ell }^{1+2\delta + (\ell \wedge 2) d} \right) \sc{t_k}^{1+2\delta} \\ &\qquad \times  \int_{\IR^d} (G_{n,t}\ast \abs f^2) (x) (G_{n,t}^{\ast(k-1)}\ast \abs g^2)(x)\d x .\label{eq:Sk estimate}
	\end{align}
	In order to estimate the time integrals in \eqref{eq:Sk estimate}, we 
	invoke \cref{lem:basic integral estimate} and find
	\begin{align}
		\nonumber
		&\left(\prod_{\ell=1}^k  t_{\ell-1} \int_0^{t_{\ell-1}} \d t_\ell \sc{ t_{\ell-1} - t_\ell }^{1+2\delta+ (\ell \wedge 2) d} \right) \sc{t_k}^{1+2\delta} \\
		\nonumber
		&\leq  \left(\prod_{\ell=1}^{k-1}  t_{\ell-1}  \int_0^{t_{\ell-1}} \d t_\ell \sc{ t_{\ell-1} - t_\ell }^{1+2\delta+ (\ell \wedge 2) d} \right) \frac{\sc {t_{k-1}}^{2(1+2\delta) + 2d}}{1} t_{k-1}^{1}  \\
		\nonumber
		&\leq  \left(\prod_{\ell=1}^{k-2}  t_{\ell-1}  \int_0^{t_{\ell-1}} \d t_\ell \sc{ t_{\ell-1} - t_\ell }^{1+2\delta + (\ell \wedge 2) d} \right) \frac{\sc {t_{k-1}}^{3(1+2\delta) + 4d}}{1\cdot 3} t_{k-2}^{3}  \\
		&\leq \cdots \leq \frac{\sc{t}^{(1+2\delta)(k+1) + d(2k-1)}  t^{2k-1} }{1 \cdot 3 \cdots (2k-1)}. \label{eq:time estimate}
	\end{align}
	For the spatial integral in \eqref{eq:Sk estimate}, we use \cref{lem:simconv} $k$ times to obtain
	\begin{align}
		\int_{\IR^d} (G_{n,t}\ast \abs f^2) (x) (G_{n,t}^{\ast(k-1)}\ast \abs g^2)(x)\d x &= \int_{\IR^d}\abs{f(x)}^2 (G_n^{\ast k}\ast \abs g^2 )(x) \d x \nonumber\\
		&\leq C_n^k \braket t^{kd} \int_{\IR^d}\abs{f(x)}^2 (G_n \ast \abs g^2 )(x) \d x,
		\label{eq:spatial estimate}
	\end{align}
	where $C_n$ denotes the same solely $n$-dependent constant as before. 
	Inserting \cref{eq:time estimate,eq:spatial estimate} into \cref{eq:Sk estimate} proves the statement.
\end{proof}
\begin{proof}[Proof of \subcref{i:remainder term estimate}]
	Using the trivial estimate $F_s^\Lambda(\ph_x,g)\le 6 \nn \ph \nn g$ and otherwise proceeding as in the previous proof, we find
	\begin{align*}
		R_N(t,f,g)  &\leq (\Cob{0} C_{n,W})^{N} (C_\ph)^{N-1}   
		\left( \int_0^{t} \d t_{1} \int_{\IR^d}\d x_1 2t \braket{t-t_1}^{1+2\delta + d}   G_{n,t-t_1} * \abs f^2(x_1) \right) 
		\\ &\qquad \times \left(\prod_{\ell=2}^{N} \int_0^{t_{\ell-1}} \d t_{\ell} \int_{\IR^d}\d x_{\ell} 2t_{\ell-1} \braket{t_{\ell-1}-t_{\ell}}^{1+2\delta + 2d}   G_{n,t_{\ell-1}-t_{\ell}}(x_{\ell-1} - x_\ell)  \right) 
		36 \nn g^2 \\
		& \leq  36 \nn \ph^2 \nn g^2(\Cob{0} C_{n,W})^{N} (C_\ph)^{N-1}  \nn \ph^{4N}   2^N \left(\prod_{\ell=1}^N  t_{\ell-1} \int_0^{t_{\ell-1}} \d t_\ell \sc{ t_{\ell-1} - t_\ell }^{1+2\delta + (\ell \wedge 2) d} \right) \\ &\qquad \times \sc{t_N}^{1+2\delta}  \int_{\IR^d} (G_{n,t}\ast \abs f^2) (x) \d x.
	\end{align*}
	Together with \eqref{eq:time estimate} this shows the claim.
\end{proof}
Combining these observations, we can now prove the main result of this \lcnamecref{sec:manybody}.
\begin{proof}[\textbf{Proof of \cref{thm:manybody}}]
	\cref{i:remainder term estimate} yields that $R_N(t,f,g) \to 0$ as $N \to \infty$, i.e., by the iteration \eqref{eq:iteration}, we find
	\[
	F_t^\Lambda(f,g)^2  \leq  \sum_{k=0}^\infty S_{k+1}(t,f,g).
	\]
	The bound of the main term \cref{i:main term estimate} can be then written and further estimated as
	\begin{align*}
		S_{k+1}(t,f,g) &\leq (\Cob{0}  )^{k+2} C_{n,W}^{k+1} C_\ph^{k} \nn \ph^{4k} 2^{k+1}   \frac{\sc{t}^{(1+2\delta)(k+2) + d(3k+2)}  t^{2k+1} }{1 \cdot 3 \cdots (2k+1)} \int_{\IR^d} \abs{f(x)}^2 (G_{n,t} \ast \abs{g}^2)(x)  \d x \\
		&\leq (\Cob{0} )^2  C_{n,W}  t  \braket{t}^{2(1+2\delta+d)} \frac{2 (\Cob{0}  C_{n,W} C_\ph \nn \ph^4  \braket{t}^{1+2\delta+3d} t^2  )^k}{ k! } \int_{\IR^d} \abs{f(x)}^2 (G_{n,t} \ast \abs{g}^2)(x)  \d x.
	\end{align*}
	Summing up these estimates proves the general statement. The case of Gaussians follows from $\nn{\ph^\sigma} \lesssim \sigma^{-d/2}$ and the bound on $\Cob0$ in \cref{prop:LRint} and \cref{rem:gaussdecay}.
\end{proof}


\section{Applications}\label{sec:applications}

In this \lcnamecref{sec:applications}, we derive the applications of \cref{thm:manybody} presented in \cref{sec:results}.

\subsection{Infinite Volume Limit}

As a first application,
we mimic the proof of \cite{GebertNachtergaeleReschkeSims.2020} for the existence of an infinite volume dynamics.
\begin{proof}[\textbf{Proof of \cref{th:infinite volume dynamics}}]
	First assume that $f$ is compactly supported. Furthermore, notice that for any $k \geq l$,
	\begin{align}
		\label{eq:pert operator}
		H_{\Lambda_k} - H_{\Lambda_l} = 	W_{\Lambda_k} - W_{\Lambda_l} = \int_{\Lambda_k \times \Lambda_k \setminus \Lambda_l \times \Lambda_l} W(x-y) \ad(\ph_x) \ad(\ph_y ) a(\ph_y) a(\ph_x ) \d(x,y)
	\end{align}
	is bounded. Therefore, we can perceive the dynamics induced by $\tau_t^{\Lambda_k}$  as a perturbation of $\tau_t^{\Lambda_l}$ with the bounded perturbation \eqref{eq:pert operator}. Thus,  the generating integral equation for the corresponding Dyson series, cf. the proof in \cite[Prop. 5.4.1]{BratteliRobinson.1996}, reads
	\begin{align}
		\label{eq:tau difference}
		\tau_t^{\Lambda_k}(a(f)) = \tau_t^{\Lambda_l}(a(f)) + \i \int_0^t \tau_s^{\Lambda_k}\left([ W_{\Lambda_k} - W_{\Lambda_l}, \tau^{\Lambda_l}_{t-s}(a(f)) ]\right)  \d s .
	\end{align}
	For simplicity assume that $\nn \ph = 1$. From \cref{prop:LRint,thm:manybody} we infer that for $n = n_W \wedge \frac{n_V}{2}$,
	\begin{align}
		\nn{\{ a^\#(\ph_x ) , \tau^{\Lambda_l}_{t}(a(f)) \}}^2 &\leq
		2 F^\Lambda_t( \ph_{x},f)^2 + 2 \nn{\{  a^\#(f), \ttzs( a^\# (\ph_{x}) ) \}}^2 \nonumber
		\\
		&=
		2 F^\Lambda_t( \ph_{x},f)^2 + 2 \nn{\{  a^\#(f),  a^\# ( \e^{\pm \i t \dG(T)} \ph_{x})\}}^2 \nonumber
		\\ 
		&\leq 2( \Cob{0} 
		\braket{t}^{1+2\delta} + \Ximb(t) )   G_{n,t}(d(\supp f,x))  \nn{f}^2.\label{eq:anti commutator estimate}
	\end{align}
	Inserting \eqref{eq:pert operator} into \eqref{eq:tau difference}, using the commutator product rule and \eqref{eq:comm anticomm} to expand the commutators as in the proof of \cref{cor:commutators}, subsequently estimating with \eqref{eq:anti commutator estimate} and using the symmetry of $W$, we finally arrive at
	\begin{align*}
		&\nn{\tau_t^{\Lambda_k}(a(f)) - \tau_t^{\Lambda_l}(a(f)) } \leq \abs t \nn{ [ W_{\Lambda_k} - W_{\Lambda_l}, \tau^{\Lambda_l}_{t-s}(a(f)) ] } \\
		&\leq \abs t \int_{\Lambda_k \times \Lambda_k \setminus \Lambda_l \times \Lambda_l} \abs{ W(x-y) } \nn{ [ \ad(\ph_x) \ad(\ph_y ) a(\ph_y) a(\ph_x ) , \tau^{\Lambda_l}_{t-s}(a(f))] }  \d(x,y) \\
		&\leq \abs t \int_{\Lambda_k \times \Lambda_k \setminus \Lambda_l \times \Lambda_l} \abs{ W(x-y) } \bigg( \nn{\{ a(\ph_x ) , \tau^{\Lambda_l}_{t-s}(a(f)) \}} + \nn{ \{ 	a(\ph_y ) , \tau^{\Lambda_l}_{t-s}(a(f)) \} }   \\ &\qquad + \nn{\{ 	\ad(\ph_y ) , \tau^{\Lambda_l}_{t-s}(a(f)) \}} + \nn{\{ 	\ad(\ph_x ) , \tau^{\Lambda_l}_{t-s}(a(f)) \}} \bigg) \d(x,y) \\
		&\leq 4 \sqrt{ 2( \Cob{0} 
			\braket{t}^{1+2\delta} + \Ximb(t) ) }   \abs t \nn{f}  \int_{\Lambda_k \times \Lambda_k \setminus \Lambda_l \times \Lambda_l}  \abs{ W(x-y) }  G_{n,t}(\dist(\supp f,x))^{1/2}  \d(x,y) .
	\end{align*}
	Now, since $n \wedge n_W > 2d$ by assumption, we have $W \in L^1(\IR^d)$ and $G_{n,t}^{1/2} \in L^1(\IR^d)$, so the integrand in the last integral is in $L^1(\IR^d \times \IR^d)$ and the bound is uniform for bounded $\abs t$. This shows that $(\tau_t^{\Lambda_k}(a(f)) )_k$ is indeed a Cauchy sequence uniformly for $t$ in a bounded interval. One can then exactly proceed as in  \cite[pp.~3629-3630]{GebertNachtergaeleReschkeSims.2020} to finish the argument.
\end{proof}

\subsection{Clustering}
Let us now prove our clustering result, adapting the method introduced in \cite{NachtergaeleSims.2006} to the continuous setting.

\begin{proof}[\textbf{Proof of \cref{th:exponential clustering}}]
	Let $\gs$ be the ground state of $H_\Lambda^\sigma$ with eigenvalue $\gse$ and eigenprojection $\gsP$.	By the same abstract reasoning as in \cite[p.~125]{NachtergaeleSims.2006}, we find 
	\begin{align*}
		\abs{\sc{\gs,A \tau_{\i b}(B) \gs}} \leq \limsup_{L\to\infty} \Abs{\frac{1}{2\pi \i} \int_{-L}^L \frac{\sc{\gs,A \tau_{t}(B) \gs}}{t - \i b}  \d t}.
	\end{align*}
	for any observable $B$ with $\gsP B \gs = 0$.
	We can then estimate
	\[
	\Abs{\frac{1}{2\pi \i} \int_{-L}^L \frac{\sc{\gs,A \tau_{t}(B) \gs}}{t - \i b}  \d t} \leq e^{-\alpha b^2}( I_1 + I_2^\pm + I_3 ),
	\]
	where
	\begin{align*}
		I_1 &:= \Abs{- \frac{1}{2\pi \i} \int_{-L}^L \frac{ \sc{\gs, \tau_{t}(B) A \gs} e^{-\alpha t^2} }{t - \i b}  \d t}, \\
		I_2^\pm &:= \Abs{\frac{1}{2\pi \i} \int_{-L}^L \frac{\sc{\gs,[ A, \tau_{t}(B)]_\pm \gs } e^{-\alpha t^2} }{t - \i b}  \d t}, \\
		I_3 &:= \Abs{ - \frac{1}{2\pi \i} \int_{-L}^L \frac{\sc{\gs,A \tau_{t}(B) \gs} (e^{\alpha b^2} - e^{-\alpha t^2})  }{t - \i b}  \d t},
	\end{align*}
	for both the commutator $[\cdot,\cdot]_- = [\cdot,\cdot]$ or the anticommutator  $[\cdot,\cdot]_+ = \{\cdot,\cdot\}$. 
	From \cite[eq. (40) and (47)]{NachtergaeleSims.2006} it directly follows that
	\begin{align}
		I_1 , I_3 \leq \frac{\nn A \nn B}{2} \exp\left(-\frac{\gamma^2}{4\alpha}\right).
	\end{align}
	In order to estimate $I_2^\pm$, it suffices by linearity and normalization, to consider normalized monomials in the creation and annihilation operators, so we can assume without loss of generality that $A = a_1 \ldots a_N$, and $B = b_1$ or $B = b_1 b_2$, where $a_i = a^\#(f_i)$, $f_i \in L^2(\IR^d)$ with $\nn{f_i} = 1$ and $b_j = a^\#(\ph_{x_j})$. In fact, let us restrict to the case $N > 1$ and $B = b_1 b_2$, since the other cases are easier with the only difference that one has to use the commutator-anticommutator expansion \cref{eq:comm anticomm} instead of the commutator expansion in the following.
	We have
	\begin{align*}
		[A, \ttls(B)] &= [A,\ttls(b_1)]\ttls(b_2) + \ttls(b_1)  [A,\ttls(b_2)] = I_{2,1} + I_{2,2},
	\end{align*}
	where
	\begin{align*}
		I_{2,1} &:= [A,\ttls(b_1) - b_1]\ttls(b_2)   + \ttls(b_1)  [A,\ttls(b_2) - b_2], \\ 	I_{2,2} &:=  [A,b_1]\ttls(b_2) + \ttls(b_1)  [A, b_2].
	\end{align*}
	First we treat $I_{2,1}$ and start with 
	\begin{align*}
		\nn{I_{2,1}} &\leq \sum_{j=1}^2 \left( \nn{[A,\ttls(b_j) - \ttzs(b_j)]} + \nn{[A,\ttzs(b_j) - b_j]} \right).
	\end{align*}
	To estimate both terms, we use \cref{cor:commutators,prop:LRint}, 
	\begin{align}
		\nn{[A,\ttls(b_j) - \ttzs(b_j)]}^2  &\leq N	\Ximb(t) \sum_{i=1}^N \int  G_{n,t}(x-y) \abs{f_i(x)}^2 \abs{\ph_{x_j}(y)}^2 \d(x,y),  \label{eq:int1} \\
		\nn{[A,\ttzs(b_j) - b_j]}^2 &\leq \Cob{1} N \abs t^2 \braket{t}^{2\delta +1} \sum_{i=1}^{N} \int G_{n,t}(x_j - y) \abs{f_i(y)}^2 \d y. \label{eq:int2} 
	\end{align}
	Since $\abs{x-y} \geq \dmin$ and $\abs{x_j - y} \geq \dmin$ in both integrals in \eqref{eq:int1} and \eqref{eq:int2}, and because of the normalization, both integrals can be estimated by 
	\begin{align*}
		\left( 1 \wedge \frac{\braket t}{\dmin} \right)^n .
	\end{align*}
	Thus, there exists $C > 0$ such that for all $s > 0$,
	\begin{align}
		\int_{-s}^{s}  \frac{ \abs{ \sc{\gs,I_{2,1} \gs} } e^{-\alpha t^2} }{ \abs{t}}  \d t &\leq  C N \int_{-s}^{s} \left(1 \wedge \frac{\braket t}{\dmin} \right)^{n/2}  \left( \abs t^{-1/2} \sqrt{\frac{\Ximb(t)}{\abs t} } +  \braket t^{\delta + 1/2} \right) \d t \nonumber \\
		&\leq  C N  \left(1 \wedge \frac{\braket s}{\dmin} \right)^{n/2}  \left(  4 \sqrt s \sqrt{\frac{\Ximb(s)}{\abs s} }   +   2s \braket s^{\delta + 1/2}  \right) . \label{eq:ss integral}
	\end{align}
	Furthermore, for the part $\abs t \geq s$ of the integral we use the same bound as in \cite{NachtergaeleSims.2006},
	\begin{align}
		\int_{\abs t \geq s}  \frac{ e^{-\alpha t^2}}{\abs t} \d t 
		\leq \frac{2}{s} \int_s^\infty e^{-\alpha t^2} \d t
		= \frac{2}{s} \int_0^\infty e^{-\alpha(t+ s)^2} \d t 
		\leq   \frac{2 e^{-\alpha s^2}}{s} \int_0^\infty e^{-\alpha t^2 } \d t  = \frac{\sqrt \pi}{s \sqrt \alpha} e^{- \alpha s^2}.
		\label{eq:tail estimate}
	\end{align}
	Combining \eqref{eq:ss integral} and \eqref{eq:tail estimate}, and using the trivial norm bound $\nn{I_{2,1}} \leq 8$, we get for another constant $C > 0$ and all $s > 0$,
	\begin{align}
		\nonumber
		&\Abs{\frac{1}{2\pi \i} \int \frac{\sc{\gs, I_{2,1} \gs } e^{-\alpha t^2} }{t - \i b}  \d t}  \leq  \frac{1}{2\pi} \int_{-s}^{s}  \frac{ \abs{ \sc{\gs,I_{2,1} \gs} } e^{-\alpha t^2} }{ \abs{t}}  \d t + \frac{1}{2\pi} \int_{\abs t \geq s}  \frac{8}{\abs t} e^{-\alpha t^2} \d t \\
		&\leq C \left(  N  \left(1 \wedge \frac{\braket s}{\dmin} \right)^{n/2}  \left(  4 \sqrt s \sqrt{\frac{\Ximb(s)}{\abs s} }   +   2s \braket s^{\delta + 1/2}  \right)  + \frac{1}{s \sqrt \alpha} e^{- \alpha s^2} \right). 	\label{eq:I21 est}
	\end{align}
	It remains to estimate the term involving $I_{2,2}$. 
	Since zero is the ground state energy, we have
	\begin{align*}
		\sc{\gs, [A,b_1] \ttls(b_2)   \gs } =\sum_{i=1}^N [a_i, b_1]  \sc{ \gs,   \widehat{A_i}  e^{\i H^\sigma_\Lambda  t } b_2 \gs },
	\end{align*}
	where $\widehat{A_i} := a_1 \ldots a_{i-1} a_{i+1} \ldots a_N$,
	With the methods in \cite{NachtergaeleSims.2006} (cf. Lemma 1 or eq. (40) with $\gamma = 0$), we can estimate  
	\begin{align*}
		\limsup_{L\to\infty} \Abs{ \frac{1}{2\pi \i} \int_{-L}^{L} \sc{ \gs,   \widehat{A_i}  e^{\i H^\sigma_\Lambda  t } b_2 \gs } } \leq \frac{1}{2}.
	\end{align*}
	Doing the same with the second term of $I_{2,2}$, we arrive at
	\begin{align}
		\label{eq:I22 est}
		\limsup_{L\to\infty} \Abs{ \frac{1}{2\pi \i} \int_{-L}^{L} \sc{ \gs,  I_{2,2} \gs } \d t } \leq \frac{1}{2} \sum_{i=1}^{N} \left( \abs{\sc{f_i,\ph_{x_1}  }} +  \abs{\sc{f_i,\ph_{x_2}  }} \right) \leq C N  \left(1 \wedge \frac{1}{\dmin} \right)^{n/2} 
	\end{align}
	for some constant $C > 0$.
	Thus, using \eqref{eq:I21 est} and \eqref{eq:I22 est},  choosing some $\delta \in (0,1)$, we obtain
	\begin{align}
		I^\pm_2 &\leq C \left(   N  \left(1 \wedge \frac{\braket s}{\dmin} \right)^{n/2}  \left(  4 \sqrt s \sqrt{\frac{\Ximb(s)}{\abs s} }   +   2s \braket s^{\delta + 1/2}  \right)  + \frac{1}{s \sqrt \alpha} e^{- \alpha s^2} \right) \\
		&\leq C' \left(  N \left(1 \wedge \frac{1}{\dmin} \right)^{n/2}  \exp(  s^{3+3d} ) + \frac{1}{s \sqrt \alpha} e^{- \alpha s^2}  \right)
	\end{align}
	for constants $C,C' > 0$. 
	Now we set $\alpha = \frac{\gamma}{2s}$ and $s^{3+3d} = \frac{n}{4} \log( \dmin+1)$. Then there exists a constant $C$ such that for all $\dmin >0$,
	\begin{align*}
		I_2^\pm \leq  \frac{C N}{\sqrt \gamma \wedge 1} \exp  \left( - \frac{\gamma \wedge 1 }{2}    \left(\frac{n}{4} \log (\dmin+1) \right)^{1/(3+3d)} \right).
	\end{align*}
	This shows the claim. 
	%
	%
\end{proof}
\subsection{Conditional Expectation}
\label{sec:conditional expectation}
Now, we introduce the conditional expectation announced in the end of \cref{sec:results}.

Let $\hk \subseteq L^2(\IR^d)$ be a closed subspace and let $(\fXc{n})$ be an orthonormal basis of $\hk^\perp$.  
We define the conditional expectation in its Kraus representation. To this end, we  introduce the following unitary operators
\begin{align*}
	u_n^{(0)} = \Id, \quad u_n^{(1)} = \ad(\fXc{n}) + a(\fXc{n}), \quad  u_n^{(2)} = \ad(\fXc{n}) - a(\fXc{n}), \quad u_n^{(3)} = \Id - 2 \ad(\fXc{n}) a(\fXc{n}).
\end{align*}
For $\underline{\alpha} = (\alpha_1, \ldots, \alpha_N) \in \{0,1,2,3\}^N$, let 
\[
u(\underline \alpha) := u_1^{(\alpha_1)} \cdots u_n^{(\alpha_n)}.
\]
For any $A \in \Af$, the CAR algebra over $L^2(\IR^d)$, and $N \in \IN$ we write
\[
\IE_\hk^N(A) := \frac{1}{4^N} \sum_{\underline \alpha \in  \{0,1,2,3\}^N}   u(\underline \alpha)^* A u(\underline\alpha).
\]
\begin{ex}
	\label{ex:L2X}
	The primary example we have in mind for the choice of $\hk$ is to obtain a conditional expectation with respect to a region $X \subseteq \IR^d$, similarly to the discrete setting in \cite{NachtergaeleSimsYoung.2018}. This means, for a measurable $X \subseteq \IR^d$ we consider $\hk = L^2(X) \subseteq L^2(\IR^d)$.
\end{ex}
\begin{lem}
	The limit $\lim_{N\to\infty}   \IE_\hk^N(A)$ exists in operator norm. 
\end{lem}
\begin{proof}
	First notice that for any monomial $A = a^\#(g_1) \ldots a^\#(g_k)$ with an odd number $k$ of creation and annihilation operators in elements $g_j \in \hk$, we have
	\[
	\sum_{\alpha = 0}^3 (u_n^{(\alpha)})^* A u_n^{(\alpha)} = 0,
	\]
	since $A$ commutes with $u_n^{(0)}$ and $u_n^{(3)}$ and anti-commutes with the other two. 
	Similarly, we have 
	\[
	\frac{1}{4} \sum_{\alpha = 0}^3 (u_n^{(\alpha)})^* A u_n^{(\alpha)} = A
	\]
	when $k$ is even. 
	Furthermore, if $\{g_1, \ldots, g_k \} \subseteq \{\fXc{1}, \ldots, \fXc{N}\}$ with $N > k$, we see by the same argument that
	\[
	\frac{1}{4^N}  \sum_{\underline \alpha \in  \{0,1,2,3\}^N}   u(\underline \alpha)^* A u(\underline\alpha) = \begin{cases}
		\frac{1}{4^k}  \sum_{\underline \alpha \in  \{0,1,2,3\}^k}   u(\underline \alpha)^* A u(\underline\alpha) &: k \text{ is even, } \\
		0 &: k \text{ is odd.}
	\end{cases}
	\]
	This shows that for any monomial $P$ of creation and annihilation operators in an ONB of $L^2(\IR^d)$, there exists a $N_P$ such that for $N \geq N_P$, we have $\IE_\hk^N(P) = \IE_\hk^{N_P}(P)$.
	
	Now let $A \in \Af$. Then there exists a sequence of polynomials in creation and annihilation operators such that $A_m \to A$. Let $\eps > 0$ and choose $m_0$ such that $\nn{A - A_m} < \eps$ for $m \geq m_0$. Write $A_m$ as a finite sum of monomials $P_1, \ldots, P_k$ and let $N_0 \geq N_{P_1}, \ldots, N_{P_k}$. Then by linearity of $\IE_\hk$, we have
	\begin{align}
		\label{eq:ev constant}
		\IE_\hk^N(A_m) =  \IE_\hk^{N_0}(A_m)
	\end{align}
	for all $N \geq N_0$. For $N,M \geq N_0$, we can write
	\[
	\IE_\hk^N(A) - \IE_\hk^M(A) =   \IE_\hk^N(A - A_m) - \IE_\hk^M(A - A_m) + \IE_\hk^N(A_m) - \IE_\hk^M(A_m).
	\]
	Then the last difference vanishes because of \eqref{eq:ev constant} and the first ones are bounded by $2 \eps$. This shows that $(\IE_\hk^N(A))_N$ is a Cauchy sequence in operator norm and therefore finishes the proof. 
\end{proof}
Let $(g_n)$ be an ONB of $L^2(\IR^d)$. 
On $\Af$ we have a unique quasi-free tracial state given by
\begin{align*}
	\omega^\tr( \ad(g_{i_n}) \ldots \ad(g_{i_1}) a(g_{i_1}) \ldots a(g_{i_n}) ) &= \frac{1}{2^n}, \\
	\omega^\tr( \ad(g_{i_m}) \ldots \ad(g_{i_1}) a(g_{j_1}) \ldots a(g_{j_n}) ) &= 0 \qquad  \text{ for } \{ i_1, \ldots, i_m \} \not= \{ j_1, \ldots, j_n \}.
\end{align*}
Let us consider the $C^*$-subalgebras of $\Af$ given by
\begin{align*}
	\Af_{\hk^\perp} = C^*( \{ a(\fXc{n}) : n \in \IN \} ), \qquad \Af_\hk = C^*( \{ a(g) : g \in \hk \} ),
\end{align*}
as well as the $C^*$-subalgebras $\Af^\pm$ and $\Af^\pm_\hk$ of $\Af$ and $\Af_\hk$, respectively, generated by monomials consisting of an {even} (+) and odd (-) number of elements.

We have the following properties of our conditional expectation.
\begin{lem}
	\label{th:tomiyama prep}
	We have
	\begin{enumlem}
		\item\label{th:tomiyama prep:1} $\nn{\IE_\hk} = 1$,
		\item\label{th:tomiyama prep:2} \label{it:IE of Xc}
		$\IE_\hk(A) = \omega^\tr(A) \Id$ 	for all $A \in \Af_{\hk^\perp}$, 
		\item\label{th:tomiyama prep:3} $(\IE_\hk)^2 = \IE_\hk$,
		\item\label{th:tomiyama prep:4} $\IE_\hk(\Af^+) = \Af^+_\hk$. 
	\end{enumlem}
	In particular, $\IE_\hk$ acts as a norm one projection from the $C^*$-algebra $\Af^+$ to the $C^*$-subalgebra $\Af^+_\hk$.
\end{lem}
\begin{proof}[Proof of \subcref{th:tomiyama prep:1}]
	This follows from $\nn{\IE_\hk^N(A)} \leq \nn A$ and $\IE_\hk^N \Id = \Id$.
\end{proof}
\begin{proof}[Proof of \subcref{th:tomiyama prep:2}]
	It suffices to show
	\begin{align*}
		\IE_\hk( \ad(\fXc{i_n}) \ldots \ad(\fXc{i_1}) a(\fXc{i_1}) \ldots a(\fXc{i_n}) ) &= \frac{1}{2^n}, \\
		\IE_\hk( \ad(\fXc{i_m}) \ldots \ad(\fXc{i_1}) a(\fXc{j_1}) \ldots a(\fXc{j_n}) ) &= 0 \qquad  \text{ for } \{ i_1, \ldots, i_m \} \not= \{ j_1, \ldots, j_n \},
	\end{align*}
	since the span of these elements lies dense in $\Af_{\hk^\perp}$ and both $\IE_\hk$ and $\omega^\tr$ are continuous. 
	The first equation follows from
	\[
	\frac{1}{4} \sum_{\alpha =0}^3  (u^{(\alpha)}_n)^* \ad(\fXc n) a(\fXc n) u^{(\alpha)}_n = \frac{1}{2}\Id
	\]
	and the fact that $u^{(\alpha)}_n$ commutes or anticommutes with $a^\#(\fXc m)$ for $n \not= m$. Similarly, the second equality holds since 	
	\begin{align*}
		\sum_{\alpha =0}^3  (u^{(\alpha)}_n)^* a^\#(\fXc n) u^{(\alpha)}_n = 0, \quad 	\sum_{\alpha =0}^3  (u^{(\alpha)}_n)^* \ad(\fXc n) a(\fXc m) u^{(\alpha)}_n = 0 \qquad \text{for } n\not= m.
	\end{align*}
\end{proof}
\begin{proof}[Proof of \subcref{th:tomiyama prep:3}]
	For any bounded operator $A$ and $\underline{\alpha} \in \{0,1,2,3\}^N$, let $\Ad_N^{(\underline\alpha)}(A) := u(\underline{\alpha})^* A u(\underline{\alpha})$. Then
	\[
	\{ \Ad_N^{(\underline\alpha)} :  \underline\alpha \in \{0,1,2,3\}^N  \}
	\]
	forms a group with respect to composition. Accordingly, we find
	\[
	(\IE_\hk^N)^2 (A) = \frac{1}{(4^N)^2} \sum_{\underline\alpha, \underline\beta \in \{0,1,2,3\}^N} \Ad_N^{(\underline\alpha)} \circ  \Ad_N^{(\underline\beta)}(A) =  \frac{1}{4^N} \sum_{\underline\alpha  \in \{0,1,2,3\}^N} \Ad_N^{(\underline\alpha)}(A) = \IE_\hk^N(A).
	\]
	Taking the limit $N\to\infty$ then yields the assertion. 
\end{proof}
\begin{proof}[Proof of \subcref{th:tomiyama prep:4}]
	First notice that the span of elements of the form $A = BC$ with $B \in \Af_{\hk}$ and $C\in \Af_{\hk^\perp}$ are dense in $\Af$. Now assume that $A \in \Af^+$. Then either $B,C \in \Af^+$ or $B,C \in \Af^-$. In the latter case we find $\IE_\hk^N(A) = 0$, since
	\[
	\IE_\hk^N(A) =  B \frac{1}{4^N} \sum_{\underline \alpha \in  \{0,1,2,3\}^N} (-1)^{\Nf_{\hk^\perp}} u(\underline{\alpha})^*(-1)^{\Nf_{\hk^\perp}} C u(\underline{\alpha}),
	\]
	where $\Nf_{\hk^\perp} = \dG(\Id_{\hk^\perp})$ denotes the number of particles in $\hk^\perp$, and for each $n$,
	\begin{align*}
		\sum_{\alpha =0}^3 &  (-1)^{N_{\hk^\perp}}  (u^{(\alpha)}_n)^*  (-1)^{N_{\hk^\perp}} a^\#(\fXc n) u^{(\alpha)}_n \\ &= 	 a^\#(\fXc n) + (u^{(3)}_n)^* a^\#(\fXc n) u^{(3)}_n - \left( (u^{(1)}_n)^* a^\#(\fXc n) u^{(1)}_n  + (u^{(2)}_n)^* a^\#(\fXc{n}) u^{(2)}_n \right) = 0.
	\end{align*}
	On the other hand, if $B,C \in \Af_+$, we have
	\[
	\IE_\hk^N(A) =  B \frac{1}{4^N} \sum_{\underline \alpha \in  \{0,1,2,3\}^N}   u(\underline{\alpha})^* C u(\underline{\alpha}) = B \IE_\hk^N(C).
	\]
	Thus, taking the limit $N\to\infty$ and using \cref{it:IE of Xc} yields $\IE_\hk(A) = B \omega^\tr(C) \in \Af^+_\hk$. This shows. $\IE_\hk(\Af^+_\hk) \subseteq \Af^+_\hk$. The other inclusion follows from $\IE_\hk(A) = A$ for $A \in \Af_\hk^+$. \qedhere
\end{proof}
Given the result of \cref{th:tomiyama prep}, a theorem by Tomiyama \cite[Theorem 1]{Tomiyama.1957} now implies the property of a conditional expectation for even parity.
\begin{cor}
	\label{th:tomiyama coro}
	For all $A \in \Af^+$ and $B,C \in \Af_\hk^+$, we have
	\[
	\IE_\hk (BAC) = B \IE_\hk(A) C.
	\]
\end{cor}
Next, we want to apply the Lieb--Robinson bound to the situation as in \cref{ex:L2X}. However, in the continuum an orthonormal basis of bounded regions has infinite cardinality (so $N=\infty$ above) which creates fundamental problems. We side-step these by introducing a notion of ``partial partial trace (PPT)'' where, roughly speaking, we select for each dyadic annulus $A_i$ (in distance) a finite orthnormal family $\{f_{i,j}\}_{1\leq i\leq J_i}\subset L^2(A_i)$ whose cardinality $J_i$ grows at most exponentially in $i$ (so that it is controlled by the decay in distance coming from the LRB).
	
More precisely, this end let $f_{i,j}$, $i \in \IN_0$, $j \in \IN$, be an ONB of $L^2(X^\sfc)$ such that 
\[
\dist(  X, \supp f_{i,j} ) \geq C_X 2^i, \qquad i \in \IN,
\]
for a fixed constant $C_X > 0$. 
\begin{prop}
	\label{th:conditional expectation LR}
	Let $(f_{i,j})_{i,j}$ be an ONB given as above, and for each $i$ assume we have a $J_i \in \IN$ with $J_i \leq C_J  2^{in/4}$ for some constant $C_J \geq 0$. Set
	\[
	\hk = \operatorname{lin}\{ f_{i,j} : i \in \IN_0,~ j \leq J_i   \}^\perp
	\]
	and use the notation $\IE_X = \IE_\hk$. 
	Let $A = a_1 \ldots a_M$, $M > 1$, with $a_j = a^\#(g_j)$, where $g_j$ is a normalized Schwartz function with $\supp g_j \subseteq X$.
	Then for each $n\le\nmax$, there exists a constant $C$ such that for all $t \in \IR$, 
	\begin{align*}
		\nn{\ttls(A) - \IE_X(\ttls(A))} \leq  C C_J(  \Ximb(t)^{1/2} + 2 \sqrt{\Cob{1}}  \abs t \braket t^{1/2 + \delta}  ) \frac{1}{n} \left( \frac{\braket t}{ C_X} \right)^n.
	\end{align*}
\end{prop}
\begin{proof}
	First consider an arbitrary closed subspace $\hk$ of $L^2(X^c)$ with corresponding ONB $(f_k)$. For any $N \in \IN$, we then compute 
	\begin{align*}
		\ttls(A) - \IE_X^N(\ttls(A)) &= \frac{1}{4^N} \sum_{\underline \alpha \in  \{0,1,2,3\}^N} \left( u(\underline{\alpha})^*u(\underline{\alpha}) \ttls(A) -  u(\underline{\alpha})^* A(t 	) u(\underline{\alpha}) \right) \\
		&= \frac{1}{4^N} \sum_{\underline \alpha \in  \{0,1,2,3\}^N}  u(\underline{\alpha})^* [u(\underline{\alpha}), \ttls(A) ]  .
	\end{align*}
	Using the definitions of $u(\underline{\alpha})$, the commutator and commutator-anticommutator expansion \eqref{eq:comm anticomm}, we find 
	\begin{align}
		\nonumber \nn{\ttls(A) - \IE_X^N(\ttls(A))} &\leq \frac{1}{4^N}	\sum_{\underline \alpha \in  \{0,1,2,3\}^N}  \nn{	[ u(\underline \alpha)  , \ttls(A)] }   \leq \frac{1}{4^N} 4^{N-1} \sum_{k=1}^{N} \sum_{\alpha_k =0}^3 \nn{ [ u_k^{(\alpha_k)} , \ttls(A)] } \\
		\nonumber  &= \frac{1}{4} \sum_{k=1}^{N}  \sum_{\alpha =0}^3 \nn{ [ u_k^{(\alpha)} , \ttls(A)] }  \\
		&\leq \frac{1}{4} \sum_{k=1}^{N}  \sum_{\alpha =0}^3 \sum_{j=1}^M \nn{ \ttls(a_1 \ldots a_{j-1}) [ u_k^{(\alpha)} , \ttls(a_j)]_{\pm,\alpha} \ttls( a_{j+1} \ldots a_M)  } ,
		\label{eq:At EX difference}
	\end{align}
	where $[\cdot,\cdot]_{\pm,\alpha} := \{\cdot,\cdot \}$ for $\alpha =1,2$ and $[\cdot,\cdot]_{\pm,\alpha} := [\cdot,\cdot]$ for $\alpha = 0,3$.
	Then each summand can be bounded by
	\begin{align*}
		\nn{ [ u_k^{(\alpha)} , (\ttls-\ttzs)(a_j)]_{\pm,\alpha} } + \nn{[ u_k^{(\alpha)} , \ttzs(a_j)]_{\pm,\alpha}}.
	\end{align*}
	The square of first term can be bounded with \cref{cor:commutators}, i.e.,
	\begin{align} 
		\nonumber
		\nn{ [ u_k^{(\alpha)} , (\ttls-\ttzs)(a_j)]_{\pm,\alpha} } &\leq  2 \left(  \Ximb(t) \int_{\IR^d\times\IR^d}
		G_{n,t}(x-y) \abs{\fXc k(x)}^2\abs{g_j(y)}^2 \d(x,y) \right)^{1/2} \\
		&\leq  2 \left(  \Ximb(t) G_{n,t}(\dist(X, \supp f_k)) \right)^{1/2}. \label{eq:uk mb}
	\end{align}
	The second one can be bounded with \cref{prop:LRint} by
	\begin{align}
		\nonumber
		\nn{[ u_k^{(\alpha)} , \ttzs(a_j)]_{\pm,\alpha}} &\leq 2 \abs{\braket{ \fXc k, (e^{\i t T} - \Id) g_j }} \leq 2 \sqrt{\Cob{1}} \abs t \braket t^{1/2 + \delta}  \left( \int  G_{n,t}(x-y) \abs{f_k(y)}^2 \d y \right)^{1/2} \\ &\leq 2 \sqrt{\Cob{1}}  \abs t \braket t^{1/2 + \delta}  G_{n,t}(\dist(X, \supp f_k))^{1/2}.
		\label{eq:uk ob}
	\end{align}
	Now we consider the concrete ONB as given in the assumption. For any $I \in \IN$, we choose $N = \sum_{i=1}^I J_i$ and sort the basis elements linearly, i.e., $f_1, f_2, \ldots, f_{N_I} = f_{1,1}, f_{1,2}, \ldots, f_{1,J_1}, f_{2,1}, \ldots, f_{I, J_I}$. Then using the estimates \eqref{eq:uk mb} and \eqref{eq:uk ob} in \eqref{eq:At EX difference}, we obtain 
	\begin{align}
		\label{eq:At prefinale}
		\begin{aligned}
					&\nn{\ttls(A) - \IE_X^N(\ttls(A))} \\&\qquad\leq   M  ( 2 \Ximb(t)^{1/2} + 2 \sqrt{\Cob{1}}  \abs t \braket t^{1/2 + \delta}  ) \braket t^{n/2} \left( J_0 \dmin^{-n/2}  + \sum_{i=1}^I J_i (C_X 2^i + \dmin)^{-n/2}  \right).
		\end{aligned}
	\end{align}
	Now, by the assumption on the $J_i$, using $(a+b)^{n/2} \geq a^{n/2} + b^{n/2}$, $n \geq 2$, and the formula
	\[
	\int \frac{\xi^{s/2}}{ C \xi^s + \eta } \d s = 2 \frac{\arctan\left( \frac{\sqrt C \xi^{s/2}}{\sqrt{\eta}} \right) }{\sqrt{C\eta} \log \xi }, \qquad C,\xi,\eta > 0,
	\]
	we find 
	\begin{align*}
		\sum_{i=1}^\infty J_i (C_X 2^i + \dmin)^{-n/2}  \leq  C_J \sum_{i=1}^\infty \frac{2^{ni/4}}{ (C_X 2^i)^{n/2} + \dmin^{n/2}  } \leq  C_J \int_1^\infty \frac{ (2^{n/2})^{s/2} }{C_X^{n/2} (2^{n/2})^{s} + \dmin^{n/2}    } \d s \leq \frac{4 C_J}{ C_X^{n/4} \dmin^{n/4} n \log 2 }.
	\end{align*}
	Using this estimate in \eqref{eq:At prefinale} yields the desired estimate. 
\end{proof}


\appendix
\section{Estimates on Commutator Expansions}\label{app:commest}
In this \lcnamecref{app:commest}, we derive the bounds \cref{eq:commbounds} using the methods from \cite{IvriiSigal.1993,HunzikerSigal.2000,ArbunichPusateriSigalSoffer.2021}. More details can be found in those articles.

The following \lcnamecref{lem:commexp} is a version of \cite[Lemma~A.1]{ArbunichPusateriSigalSoffer.2021} with explicit upper bounds.
To express these, for $f\in\cC^{n+2}(\IR)$, we introduce the norm
\begin{equation}\label{def:normnp}
	\norm{f}_{n,p} \coloneqq \sum_{k=0}^{n+2}\int \braket{x}^{k-p-1}|f^{(k)}(x)|\d x \qquad \mbox{for}\ n\in\IN,\ p\in\IR, \quad \mbox{where}\ \braket{x}\coloneqq\sqrt{1+x^2}.
\end{equation}
\begin{lem}\label{lem:commexp}
	Let $H$ be defined as in \cref{prop:APSS} and
	let $f,g\in\cC_\sfb^{n+2}(\IR)$ such that $\norm{g}_{n,n/2}<\infty$ and $\norm{f}_{n,n}<\infty$.
	Then there exist bounded operators $B_{g,k}$, $k=1,\ldots,n-1$ and $R_{g,f,n}(a,s)$, $a>0,s \geq 1$ satisfying
	\begin{equation}\label{eq:commbounds}
		\norm{\wh p^j \! B_{g,k}} \le c_k\norm{g}_{k,(k-j)/2}, \quad \sup_{a>0,s\geq 1} \norm{\wh p^j \! R_{g,f,n}(a,s)}\le c_n( \norm{\wh p^j B_{g,n}}\norm{f}_{n,n} + j\norm{B_{g,n}}\norm{f}_{n,n+1}), \quad j=0,1,
	\end{equation}
	where $c_k>0$ solely depends on $k\in\IN$ (and the choice of $V$),
	such that
	\begin{equation}\label{eq:commexp}
		\big[ g(H),f(\bx as)\big] = \sum_{k=1}^{n-1}\frac{s^{-k}}{k!}B_{g,k}f^{(k)}(\bx as) + s^{-n}R_{g,f,n}(a,s),
	\end{equation}	
	where $\bx as$ is defined as in \cref{def:bx}.
\end{lem}

\begin{proof}
	The statement is exactly that of \cite[Lemma~A.1]{ArbunichPusateriSigalSoffer.2021}, whence we restrict to the proof of the upper bounds using the techniques from \cite[Appendix~B]{HunzikerSigal.2000}.
	
	To bound $B_{g,k}$, $k\in\IN$, we use the integral representation
	\[ B_{g,k} = \sum_{\pi} \int I(Z,\pi)  \d \wt g(z) \quad \mbox{with}\quad I(z,\pi)\coloneqq (z-H)^{-1}\prod_{\ell=1}^{\abs\pi}C_{\pi_i}(z-H)^{-1},  \]
	where $\wt g$ is a specific almost analytic extension of $g$ supported in a complex neighborhood of $\supp g$, cf. \cite[Eq.~(B.5)]{HunzikerSigal.2000} and $\d \wt g(z) := -(2\pi)^{-1} \partial_{\overline z} \wt g(z) \d x \d y$ with $z = x + \i y$. The sum runs over all ordered partitions $\pi=(\pi_1,\ldots,\pi_{\abs\pi})$ of $1, \ldots, k$, i.e., $\sum_{\ell=1}^{\abs\pi}\pi_\ell = k$, such that $\pi_\ell\in\{1,2\}$,  see \cite[Eq.~(A.8)]{ArbunichPusateriSigalSoffer.2021} for a derivation. Furthermore, $C_j := \operatorname{ad}^j_{\braket x} H$, the $j$-fold commutator of $H$ with $\braket x$. Now we want to apply the bounds
	\[ \norm{C_2}\le 1, \quad \norm{(z-H)^{-1}}\le \frac1{\abs{\Im z}} \quad \mbox{and}\quad \norm{\wh p (z-H)^{-1}} \vee \norm{C_1(z-H)^{-1}} \le c\frac{\braket{\Re z}^{1/2}}{\abs{\Im z}} ,  \]
	where $c>0$ is a universal constant dependent on the choice of $V$, cf. \cite[Eq.~(A.13)]{ArbunichPusateriSigalSoffer.2021}.
	Observing that the set $\{\ell\in\{1,\ldots,\abs\pi\}:\pi_\ell=1\}$ has the cardinality $2\abs\pi-k$, we find
	\[ \norm{\wh p^j I(z,\pi)}\le \frac{\braket{c \Re z}^{\abs\pi-(k-j)/2}}{\abs{\Im z}^{\abs\pi+1}} \qquad \mbox{for}\ j=0,1. \]
	We now apply the bound 	
	\[ \int  \frac{f(\Re z) }{\abs{\Im z}^{p+1}}  \abs{\d \wt g(z)} \le c_n \sum_{k=0}^{n+2} \int_{\IR} f(x)\braket{x}^{k-p-1}g^{(k)}(x) \d x, \]
	where $c_n$ solely depends on $n\in\IN$ (and the construction of $\wt g$), see \cite[Eqs.~(B.3)\&(B.6)]{HunzikerSigal.2000}.
	Hence,
	inserting the definition \cref{def:normnp},
	 we find for a constant $c_k>0$ depending on $k$
	\[ \int \norm{\wh p^j I(z,\pi)}\abs{\d \wt g(z)} \le c_k \|g\|_{k,(k-j)/2} \qquad\mbox{independent of}\ \pi, \]
	which proves the first estimate in \cref{eq:commbounds}.
	The second now easily follow using the integral representation
	\[ R_{g,f,n}(a,s) =  \int (z-\bx as)^{-1}B_{g,n}(z-\bx as)^{-n}\d \wt f(z), \]
	cf. \cite[Eq.~(A.6)]{ArbunichPusateriSigalSoffer.2021} or \cite[Eq.~(B.24)]{HunzikerSigal.2000}.
\end{proof}
Let us now briefly discuss some relevant classes of functions for which $\norm{\cdot}_{n,p}$, as defined in \cref{def:normnp}, is finite. In our bounds, we set
\begin{equation}
	\iota_{k} \coloneqq \int_{\IR} \braket{x}^{-k-1}\d x<\infty \qquad \mbox{for}\ k>0
\end{equation}
and as usually write $\norm{f}_\infty \coloneqq \sup f(\IR)$. The proof of the following lemma trivially follows by estimating $f^{(k)}$ pointwise by its supremum norm. 
\begin{lem}\label{lem:npbounds}
	Let $n\in\IN$ and let $f\in\cC^{n+2}(\IR)$ such that $f^{(k)}$ is bounded for all $k=0,\ldots,n+2$.\\
	If $p>n+2$, then
	\[\norm{f}_{n,p} \le \sum_{k=0}^{n+2} \iota_{p-k}\norm{f^{(k)}}_\infty. \]
	Further, if $f^{(m)}$ is compactly supported for some $m\in\IN$ with $m< p+1$, then
	\[ \norm{f}_{n,p} \le \sum_{k=0}^{m-1} \iota_{p-k} \norm{f^{(k)}}_\infty + \sum_{k=m}^{n+2}\norm{f^{(k)}}_\infty \norm{\chr_{\supp f^{(k)}}\braket{\cdot}^{k-p-1}}_1  .\]
\end{lem}
\begin{cor}\label{corgE}
	Let $\xi$, $\alpha$, $g_E$ be defined as in \cref{prop:APSS}. Then
	\[ \norm{g_E}_{n,p} \le \iota_{p}\norm{\xi}_\infty +  \sum_{k=1}^{\lceil p \rceil } \frac{ \norm{\xi^{(k)}}_\infty}{((\alpha-1)E)^{k-1}}
			+ \sum_{k=\lceil p \rceil + 1} ^{n+2} \left(\tfrac{\alpha}{\alpha-1}\right)^{k-1} \frac{\braket{E}^{k-p-1}}{E^{k-1}}  \norm{\xi^{(k)}}_\infty
		\qquad \mbox{for all}\ E,p>0, \]
		where $\lceil p \rceil \coloneqq\min\{k\in\IN_0:k\ge p\}$.
\end{cor}
\begin{proof}
	Apply \cref{lem:npbounds} with $f=g_E$, $m=1$ and use $\braket{\alpha E}\le \alpha\braket{E}$,
	\[  g_E^{(k)}(x)=\frac{{\xi^{(k)}((x-E)/((\alpha-1) E))}}{((\alpha-1) E) ^{k}} \quad \mbox{as well as}\quad  \norm{\chr_{\supp f^{(k)}}\braket{\cdot}^{k-p-1}}_1 \le (\alpha-1)E \braket{\alpha E} ^{0\vee (k-p-1)}. \qedhere \]
\end{proof}
\section{Decay of Iterated Convolutions}

We apply the following well-known statement on the decay of the convolution of polynomially decaying functions.
\begin{lem}\label{lem:simconv}
	Let $f,g \: \IR^d\to\IC$, $m,n > d$ and $\alpha,\beta,c_f,c_g>0$  such that $\abs {f (x)} \leq c_f (1 \vee \alpha \abs x)^{-n}$ and $\abs {g (x)} \leq c_g (1 \vee \alpha \abs x)^{-m}$ for all $x \in \IR^d$. Then there exists $c_{m,n}>0$ solely depending on $m$, $n$ and $d$ such that
	\[ \abs{f\ast g(x)}
	\leq c_{m,n}c_fc_g (\alpha \wedge \beta)^{-d} \frac{1}{(1 \vee   (\alpha \wedge \beta)\abs x)^{m\wedge n}}.
	\]
\end{lem}
\begin{proof}
A direction calculation shows
	\begin{align*}
		\abs{ f*g (x)} &\leq \int_{\abs y \geq \abs x /2} \abs{f(y)}  \abs{g(x-y)} \d y  + \int_{\abs y < \abs x /2} \abs{f(y)}  \abs{g(x-y)} \d y  \\
		&\leq c_f c_g \int_{\abs y \geq \abs x /2} \frac{1}{ (1 \vee \alpha \abs y)^n}  \frac{1}{ (1\vee \beta \abs{x-y})^m} \d y  +  c_f c_g \int_{\abs y < \abs x /2}\frac{1}{ (1 \vee \alpha \abs y)^n}  \frac{1}{ (1\vee \beta \abs{x-y})^m} \d y  \\
		&\leq  \frac{c_f c_g 2^n}{ (2 \vee \alpha \abs x)^n}  \int_{\IR^d} \frac{1}{(1 \vee \beta \abs{y})^m} \d y  +   \frac{c_f c_g 2^m}{ (2 \vee \beta \abs{x})^m} \int \frac{1}{ (1 \vee \alpha \abs y)^n}   \d y  \\
		&\leq  \frac{c_f c_g 2^n}{ (2 \vee \alpha \abs x)^n} \beta^{-d} \int_{\IR^d} \frac{1}{(1 \vee  \abs{y})^m} \d y  +   \frac{c_f c_g 2^m}{ (2 \vee \beta \abs{x})^m} \alpha^{-d} \int \frac{1}{ (1 \vee \abs y)^n}   \d y. \qedhere
	\end{align*}
\end{proof}

\section*{Acknowledgements}

We thank Heinz Siedentop and Tom Wessel for useful remarks.
BH acknowledges support by the Ministry of Culture and Science of the State of North Rhine-Westphalia within the project `PhoQC'.
The research of M.L.\ is supported by the Deutsche Forschungsgemeinschaft (DFG, German Research Foundation) through grant TRR 352--470903074.


\printbibliography

\end{document}